\documentclass[journal, onecolumn]{IEEEtran}

\usepackage[utf8]{inputenc} 
\usepackage[T1]{fontenc}
\usepackage{url}
\usepackage{ifthen}
\usepackage{cite}
\usepackage[cmex10]{amsmath} 

\usepackage{enumitem}
\usepackage{algorithm, algorithmic, float}
\usepackage[font=small]{caption}
\usepackage[dvips]{graphicx}
\usepackage{multirow}
\usepackage{colortbl}
\usepackage{tabularx}
\usepackage{color}
\usepackage{setspace}
\usepackage{mathtools}
\usepackage{bm}
\usepackage{amsfonts,amssymb,amsthm,epsfig,epstopdf,url,array}
\usepackage{dsfont}
\usepackage{subfig}
\usepackage{float}
\DeclarePairedDelimiter{\floor}{\lfloor}{\rfloor}
\DeclarePairedDelimiter{\ceil}{\lceil}{\rceil}

\newcommand{\aggregate}[2]{\underset{#2}{\operatornamewithlimits{#1\ }}}
\newcommand{\Mod}[1]{\ \mathrm{mod}\ #1}
\def\BibTeX{{\rm B\kern-.05em{\sc i\kern-.025em b}\kern-.08em
    T\kern-.1667em\lower.7ex\hbox{E}\kern-.125emX}}

\newcommand{\divides}{\mid}
\newcommand{\notdivides}{\nmid}

\begin{document}

\title{Explicit Constructions of MBR and MSR Codes \\ for Clustered Distributed Storage\\
}

\author{Jy-yong~Sohn,~\IEEEmembership{Student Member,~IEEE,}
	Beongjun~Choi,~\IEEEmembership{Student Member,~IEEE,}
	and~Jaekyun~Moon,~\IEEEmembership{Fellow,~IEEE}
	\thanks{The material in this paper was presented in part at the 2018 IEEE International Symposium on Information Theory \cite{sohn2018class}. This work is in part supported by the National Research Foundation of Korea under Grant No.2019R1I1A2A02061135. The authors are with the School of Electrical Engineering, Korea Advanced
Institute of Science and Technology, Daejeon, 34141, Republic of Korea
(e-mail: jysohn1108@kaist.ac.kr, bbzang10@kaist.ac.kr, jmoon@kaist.edu).
}
}


\maketitle

\begin{abstract}
This paper considers capacity-achieving coding for the clustered form of distributed storage that reflects practical storage networks. 
To reflect the clustered structure with limited cross-cluster communication bandwidths, nodes in the same cluster are set to communicate $\beta_I$ symbols, while nodes in other clusters can communicate $\beta_c \leq \beta_I$ symbols with one another.  
We provide two types of exact regenerating codes which achieve the capacity of clustered distributed storage: the minimum-bandwidth-regenerating (MBR) codes and the minimum-storage-regenerating (MSR) codes.
First, we construct MBR codes for general parameter settings of clustered distributed storage. 
The suggested MBR code is a generalization of an existing code proposed by Rashmi \textit{et al.}, for scenarios where storage nodes are dispersed into $L>1$ clusters.
The proposed MBR code for the $\beta_c=0$ case 
requires a much smaller field size compared to existing local MBR codes. 
Secondly, we devise MSR codes for clustered distributed storage. Focus is given on two important cases: $\epsilon=0$ and $\epsilon \in [\frac{1}{n-k}, 1]$,
where $\epsilon=\beta_c/\beta_I$ is the ratio of the available cross- to intra-cluster repair bandwidths, $n$ is the total number of distributed nodes and $k$ is the number of contact nodes in data retrieval.
The former represents the scenario where cross-cluster communication is not allowed, while the latter corresponds to the case of minimum node storage overhead. 
For $\epsilon=0$, two existing locally repairable codes are proven to be MSR codes for the clustered model. For $\epsilon \in [\frac{1}{n-k}, 1]$, existing MSR codes for the non-clustered model are applicable to clustered scenarios with a simple modification.
Finally, under the settings of $\epsilon=\frac{1}{n-k}$ and $n=kL$, an MSR code is suggested which is based on simple MDS codes and requires a smaller field size for symbols than the existing code for $L\ge3$.
\end{abstract}

\begin{IEEEkeywords}
Clustered distributed storage, regenerating codes, node repair, network coding
\end{IEEEkeywords}

\section{Introduction}
Motivated by the need to handle the data deluge in modern networks,
large-scale distributed storage systems (DSSs) are now widely deployed.
With the aid of network coding, DSSs are highly tolerant to failure events, allowing users reliable access to the stored data.
The early work of \cite{dimakis2010network} obtained a closed-form expression for \textit{capacity} $\mathcal{C}(\alpha, \gamma)$, the maximum reliably storable file size, as a function of two important system parameters: the node capacity $\alpha$ and the bandwidth $\gamma$ for regenerating a failed node. The authors of \cite{dimakis2010network} also found a fundamental trade-off relationship between $\alpha$ and $\gamma$, to satisfy $\mathcal{C}(\alpha, \gamma) = \mathcal{M}$, i.e., in reliably storing a given file with size $\mathcal{M}$.
Based on the information-theoretic analysis of DSS in \cite{dimakis2010network}, several researchers \cite{rashmi2011optimal, rashmi2009explicit, shah2012TIT, suh2011exact, cadambe2013asymptotic, ernvall2014codes, Goparaju2017TIT, ye2017explicit} 
have developed explicit network coding schemes which achieve capacity of DSSs. All these works considered homogeneous setting where each node has identical storage capacity and communication bandwidth.

However, in the real world, data centers arrange their storage devices into multiple \textit{racks}, essentially forming clusters, where the available cross-rack communication bandwidth is considerably smaller than the available intra-rack bandwidth. In an effort to reflect this practical nature of data centers, several researchers recently considered the concept of clustered topologies in DSSs \cite{sohn2016capacity, prakash2017storage, hu2017optimal, sohn2018capacity,prakash2018TIT, choi2017secure, choi2019secure}, where each rack corresponds to a cluster containing multiple nodes.
Especially in \cite{sohn2016capacity}, the authors of the present paper considered clustered DSSs with $n$ storage nodes dispersed in $L$ clusters, where each node has storage capacity $\alpha$. To reflect the difference between intra-cluster and cross-cluster bandwidths, \cite{sohn2016capacity, sohn2018capacity} use two parameters for indicating repair bandwidths: $\beta_I$ for the repair bandwidth among nodes in the same cluster and $\beta_c$ for repair bandwidth between nodes in different clusters. 
Under this setting, storage capacity $\mathcal{C}(\alpha, \beta_I, \beta_c)$ of clustered DSSs -- the maximum reliably retrievable file size via a contact of arbitrary $k<n$ nodes -- has been obtained in \cite{sohn2018capacity}. 

\subsection{Main Contributions}
This paper designs explicit coding schemes which achieve capacity $\mathcal{C}(\alpha, \beta_I, \beta_c)$ of clustered DSSs computed in \cite{sohn2018capacity}.
Mainly, we focus on two types of exact-regenerating codes, the minimum-bandwidth-regenerating (MBR) code and the minimum-storage-regenerating (MSR) code, under the setting of maximal helper nodes. Both codes achieve capacity of clustered DSSs, while the former uses the minimum repair bandwidth and the latter assumes the minimum node storage overhead.
The MBR codes suggested in this paper cover arbitrary system parameter values of $n,k,L,\beta_I$ and $\beta_c$, while different types of code constructions are proposed depending on $\epsilon=\beta_c/\beta_I$, the ratio of cross- to intra-cluster repair bandwidths.
Moreover, when $\beta_I = \beta_c$, the suggested MBR code reduces to an existing MBR code in \cite{rashmi2009explicit}, i.e., the proposed scheme can be viewed as a generalization of the code construction in \cite{rashmi2009explicit}. 
Regarding MSR codes, we consider two important scenarios of $\epsilon=0$ and $\epsilon \in [\frac{1}{n-k}, 1]$.
The former represents the system where cross-cluster
communication is not possible. The latter corresponds
to the range of $\epsilon$ values that can achieve the minimum node storage overhead of $\alpha = \mathcal{M}/k$, where $\mathcal{M}=\mathcal{C}(\alpha, \beta_I, \beta_c)$ is the file size that we want to reliably store in a clustered DSS.
When $\epsilon=0$, it is shown that appropriate application of locally repairable codes suggested in \cite{papailiopoulos2014locally,tamo2016optimal} can be used as MSR
codes for general $n,k,L$ settings with the application rule
depending on the parameter setting.
For the $\epsilon \in [\frac{1}{n-k},1]$ 
case, we first show that existing MSR codes for non-clustered DSSs, e.g., \cite{suh2011exact, rashmi2011optimal}, can be used as MSR codes for clustered DSSs with arbitrary $n,k,L$, when the repair rule is modified to adjust the given $\epsilon$ setting. 
Moreover, a simple MSR code with a small required field size is suggested for $L\ge3$ under the conditions of $\epsilon=\frac{1}{n-k}$ and $n=kL$.
The proposed MBR and MSR coding schemes can be implemented in wireless storage networks having a clustered topology or data centers with multiple racks, depending on the network bandwidth constraints and the node storage overhead constraints.


\subsection{Related Works}

The MBR codes suggested in this paper for clustered DSSs are based on the MBR codes in \cite{rashmi2009explicit, shah2012TIT} for non-clustered DSSs with $\beta_I =\beta_c$.
However, critical contributions are added in our code constructions: 1) we modify the existing works to reflect the clustered nature of distributed storage having $\beta_I \geq \beta_c$, 2) a novel mathematical analysis (which includes applying the majorization theory \cite{vaidyanathan2010signal}) on the modified codes is developed to prove the exact regeneration property and the data reconstruction property for general $n,k,L,\epsilon$ settings.
Here we note that the code proposed in the present paper can be viewed as a generalization of the MBR code in \cite{rashmi2009explicit} for application to clustered DSSs; the code proposed in Section \ref{Sec:CodeDesign_nonzero_epsilon} reduces to the code in \cite{rashmi2009explicit} by setting $\beta_I = \beta_c$. 

Considering an extreme network scenario when a failed node can contact only a limited number of survived nodes, i.e., $\beta_c = 0$ in the present paper, the concept of locally repairable codes (LRCs) are considered in various papers \cite{gopalan2012locality, papailiopoulos2014locally, kamath2013ISIT, kamath2014TIT, rawat2014optimal, silberstein2018locality}. Some of these works were on codes with local regeneration, i.e., LRCs which minimize the bandwidth required in the local repair processes. Especially, within the class of codes with local regeneration, the authors of \cite{kamath2013ISIT, kamath2014TIT} provided local MBR codes where each local code is an MBR, while the authors of \cite{kamath2014TIT, rawat2014optimal} considered local MSR codes. 
Here we note that the construction rule of local MBR codes in \cite{kamath2014TIT, kamath2013ISIT} are similar to that of the code suggested in Section \ref{Sec:CodeDesign_zero_epsilon} of the present paper. However, 
the required field size is much less for the code in Section \ref{Sec:CodeDesign_zero_epsilon} compared to the existing local MBR codes.
A detailed comparison is provided in Section \ref{Sec:comparison_with_local_MBR}.

Regarding the clustered topology of storage nodes in multi-rack data centers, there have been some research \cite{tebbi2014code, hu2017optimal, sohn2018class, chen2019explicit, sahraei2017increasing, prakash2017storage} on designing regenerating codes for DSSs with multiple clusters, but to a limited extent.
The coding scheme suggested in \cite{tebbi2014code} is well suited for multi-rack systems, but not proven to be an MBR or an MSR code. 
The authors of \cite{hu2017optimal} focused on designing MSR codes 
under the assumption of maximum intra-cluster communication $\beta_I$. The focus is different from the present paper which designs MSR and MBR codes depending on given $\beta_I$ and $\beta_c$.
The authors of \cite{sahraei2017increasing} provided an explicit coding scheme which reduces the repair bandwidth of clustered DSSs under the condition that each failed node can be exactly regenerated by contacting any one of other clusters. 
The approach of \cite{sahraei2017increasing} is different from that of the present paper in the sense that it does not consider the scenario with unequal intra- and cross-cluster repair bandwidths.
Moreover, the coding scheme in \cite{sahraei2017increasing} is shown to be a minimum-bandwidth-regenerating (MBR) code for some limited parameter setting, while the present paper proposes MBR codes for arbitrary parameters.
An MSR code for clustered DSSs has been suggested in \cite{prakash2017storage}, but the data retrieval condition of \cite{prakash2017storage} is different from that of the present paper. The authors of \cite{prakash2017storage} considered the scenario where data can be collected by contacting arbitrary $k$ out of $n$ \emph{clusters}, while data is retrieved by contacting any $k$ out of $n$ \emph{nodes} in the present paper. Thus, the storage versus repair bandwidth tradeoff curves for the present paper and \cite{prakash2017storage} are different. In short, the code in \cite{prakash2017storage} and the code in the present paper achieve different tradeoff curves.
Another recent work \cite{chen2019explicit} suggested MSR codes for clustered DSS. On one hand, the authors of \cite{chen2019explicit} assumed that intra-cluster repair bandwidth does not incur any cost, and only concerned with minimizing the amount of information transmitted across different clusters.
On the other hand, the present paper considers both the intra- and cross-cluster repair burdens, and proposes MBR and MSR codes which minimize the overall repair burden for a given $\epsilon=\beta_c/\beta_I$, the ratio of cross- to intra-cluster repair bandwidth. Thus, the code in \cite{chen2019explicit} and the code in this paper achieve different tradeoff curves.
Compared to the conference version \cite{sohn2018class} of the current work, this paper 
adds an MSR code construction rule for $\epsilon \in [\frac{1}{n-k}, 1]$, and provides another MSR code for $\epsilon = \frac{1}{n-k}, n=kL$ having a much smaller field size. Moreover, MBR codes for general $n,k,L,\epsilon$ setting are added in the present paper, which were not included in \cite{sohn2018class}.

\newtheorem{theorem}{Theorem}
\newtheorem{lemma}{Lemma}
\newtheorem{corollary}{Corollary}
\newtheorem{definition}{Definition}
\newtheorem{construction}{Construction}
\newtheorem{condition}{Condition}
\newtheorem{prop}{Proposition}
\newtheorem{remark}{Remark}

\section{Problem Setup}

\subsection{Preliminaries on Clustered Distributed Storage}

\begin{figure}
	\centering
	\subfloat[][Encoding \& Distributing a file]{\includegraphics[width=65mm ]{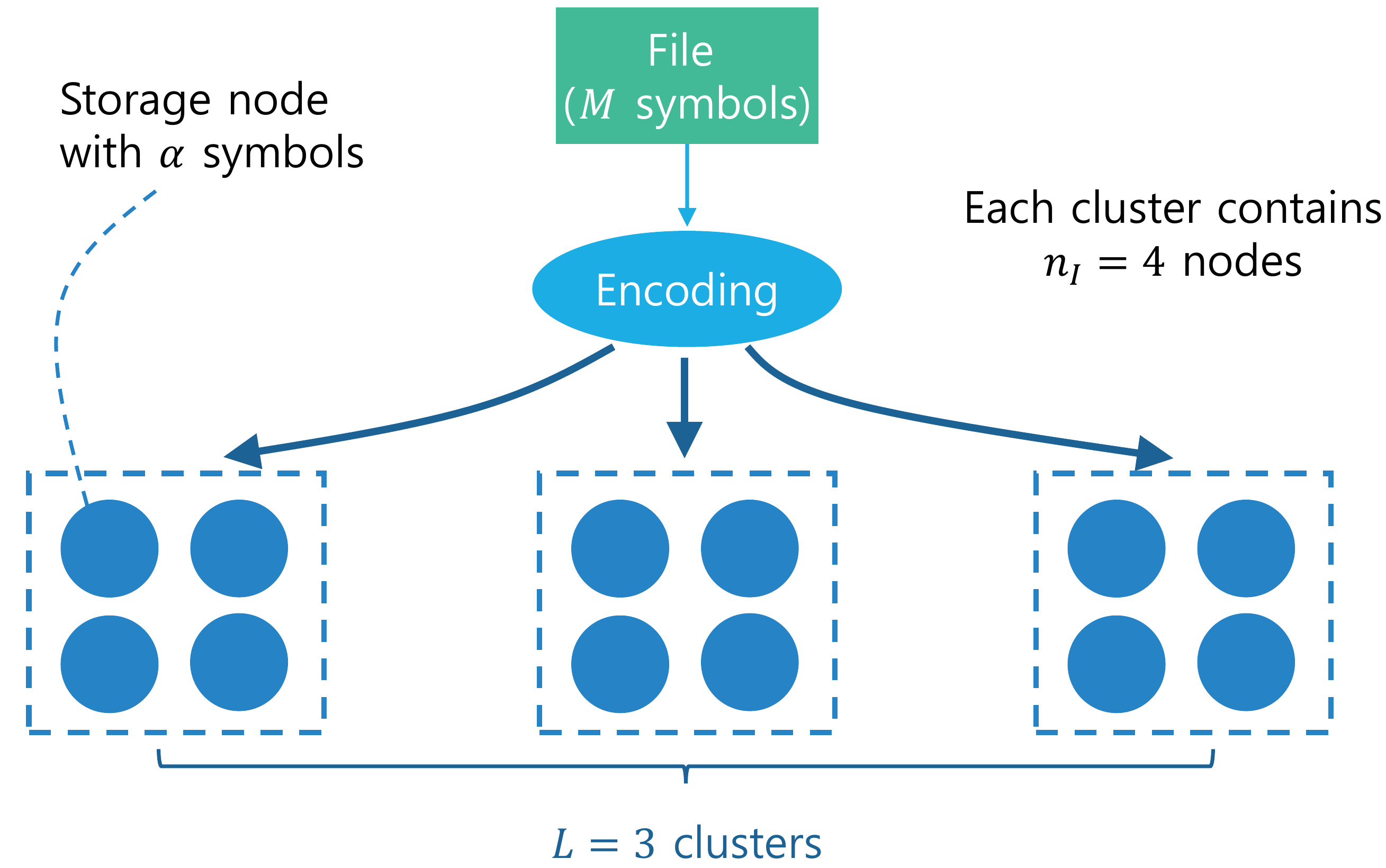}\label{Fig:Clustered_DSS}}
	\quad \quad
	\vspace{1mm}
	\subfloat[][Repairing a failed node]{\includegraphics[width=70mm]{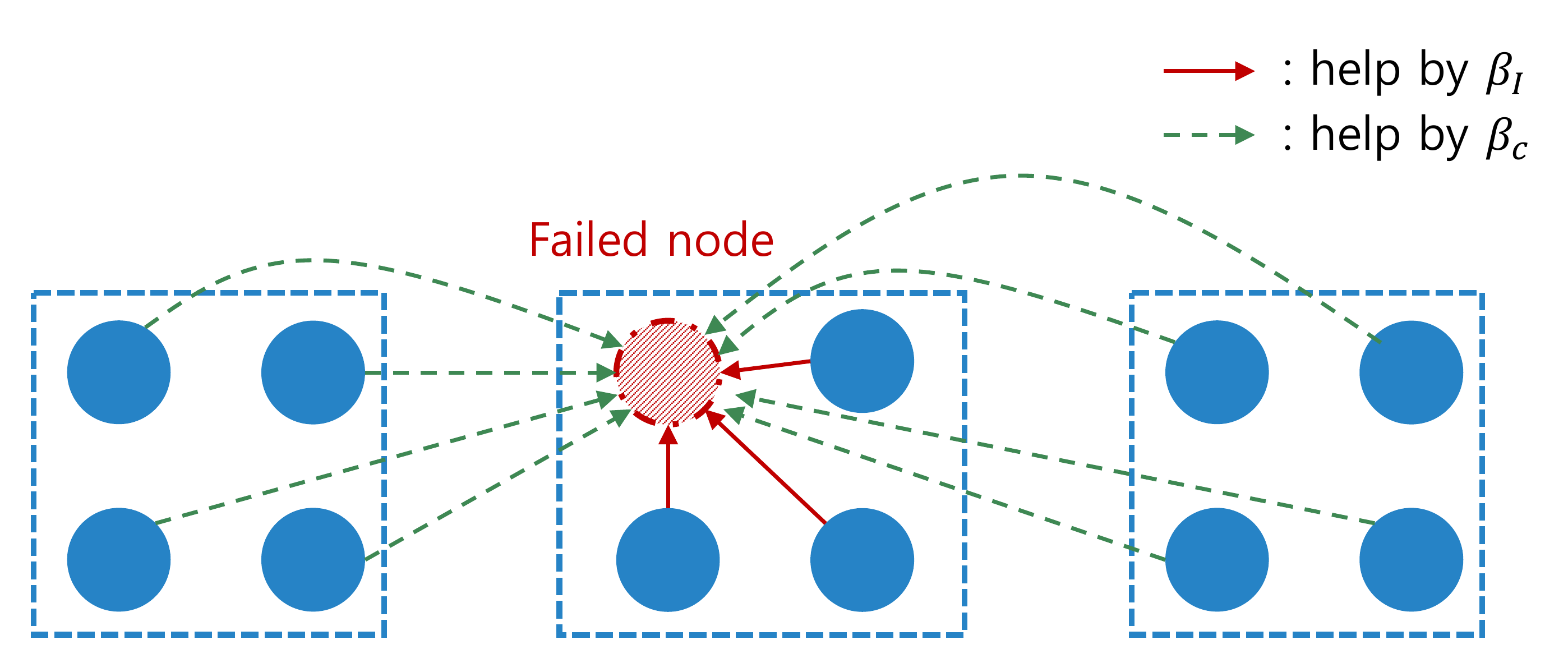}\label{Fig:repair_process}}
	\caption{Clustered DSS for $n=12, L=3, n_I=n/L = 4$. When a node fails, each helper node within the same cluster transmits $\beta_I$ symbols, while each helper node in other clusters sends $\beta_c$ symbols.}
	\label{Fig:MBR_coding_schemes}
\end{figure}

Here, we briefly summarize the clustered distributed storage system model, which is originally suggested in \cite{sohn2016capacity, sohn2018capacity}. The clustered model is motivated by the practical scenario where
data centers consist of storage nodes dispersed into multiple racks, with the cross-rack bandwidth small compared to the intra-rack bandwidth \cite{ahmad2014shufflewatcher}. However, the application is not limited to the multi-rack data centers; the clustered model can be naturally applied to general distributed storage systems (e.g. wireless sensor network) with clustered topology, where the amount of available cross-cluster bandwidth is different from that of the available intra-cluster bandwidth.  
Consider a file with $\mathcal{M}$ symbols to be stored. The file is encoded and distributed into $n$ storage nodes, which are uniformly dispersed into $L$ clusters. We use the notation $n_I=n/L$ to indicate the number of nodes in each cluster. Moreover, each node contains $\alpha$ symbols.
Fig. \ref{Fig:Clustered_DSS} provides the system model when $n=12, L=3$.

Consider the scenario where a data collector wants to retrieve the original file. Assume that it contacts arbitrary $k<n$ nodes in the system, irrespective of the cluster which contains the node.  
Moreover, when a node fails, it is regenerated by contacting $n_I-1$ survived nodes in the same cluster and $n-n_I$ nodes in other clusters. This implies that the total number of helper nodes are $n-1$, i.e., the regeneration process utilizes the maximum number of helper nodes. According to Proposition 1 of \cite{sohn2018capacity}, this maximum helper node setup is the capacity-maximizing choice. 
In the regeneration process, each node in the same cluster transmits $\beta_I$ symbols, while each node in other clusters sends $\beta_c$ symbols. Fig. \ref{Fig:repair_process} illustrates the regeneration process of a failed node.
Recall that the clustered model considers practical scenarios with limited cross-cluster communication bandwidth, compared to the abundant intra-cluster communication bandwidth. Thus, we follow the assumption $\beta_I \geq \beta_c$ used in \cite{sohn2018capacity}. Moreover, we use the parameter $\epsilon=\beta_c/\beta_I$, which is defined in \cite{sohn2018capacity} and represents the ratio of the cross-cluster repair bandwidth to the intra-cluster repair bandwidth. It is easily seen that $0 \leq \epsilon \leq 1$ holds. According to the regeneration process illustrated above, the \textit{repair bandwidth} $\gamma$ is expressed as
\begin{equation} \label{Eqn:gamma}
\gamma = (n_I-1)\beta_I + (n-n_I)\beta_c.
\end{equation}

In the clustered DSS with given parameters of $(n,k,L,\alpha, \beta_I, \beta_c)$, the authors of \cite{sohn2018capacity} defined \textit{capacity}, the maximum amount of data that is reliably retrievable by contacting arbitrary $k$ nodes.
The capacity expression for the clustered DSS is obtained in Theorem 1 of \cite{sohn2018capacity}:
\begin{equation}\label{Eqn:Capacity of clustered DSS_rev}
\mathcal{C}(\alpha, \beta_I, \beta_c)\hspace{-0.5mm}= \hspace{-1mm} \sum_{i=1}^{n_I}\hspace{-0.25mm} \sum_{j=1}^{g_i} \hspace{-0.25mm} \min \{\alpha, \hspace{-0.25mm}\rho_i\beta_I + (n-\rho_i - (\hspace{-0.25mm}\sum_{m=1}^{i-1}g_m) - j) \beta_c \},
\end{equation}
where
\begin{align}
\rho_i &= n_I - i, \nonumber\\ 
g_m &=
\begin{cases}
\lfloor k/n_I \rfloor + 1, & m \leq (k \text{ mod } n_I) \\
\lfloor k/n_I \rfloor, & otherwise.\label{Eqn:g_m}
\end{cases}
\end{align}


\begin{figure}[t]
	\centering
	\includegraphics[height=40mm]{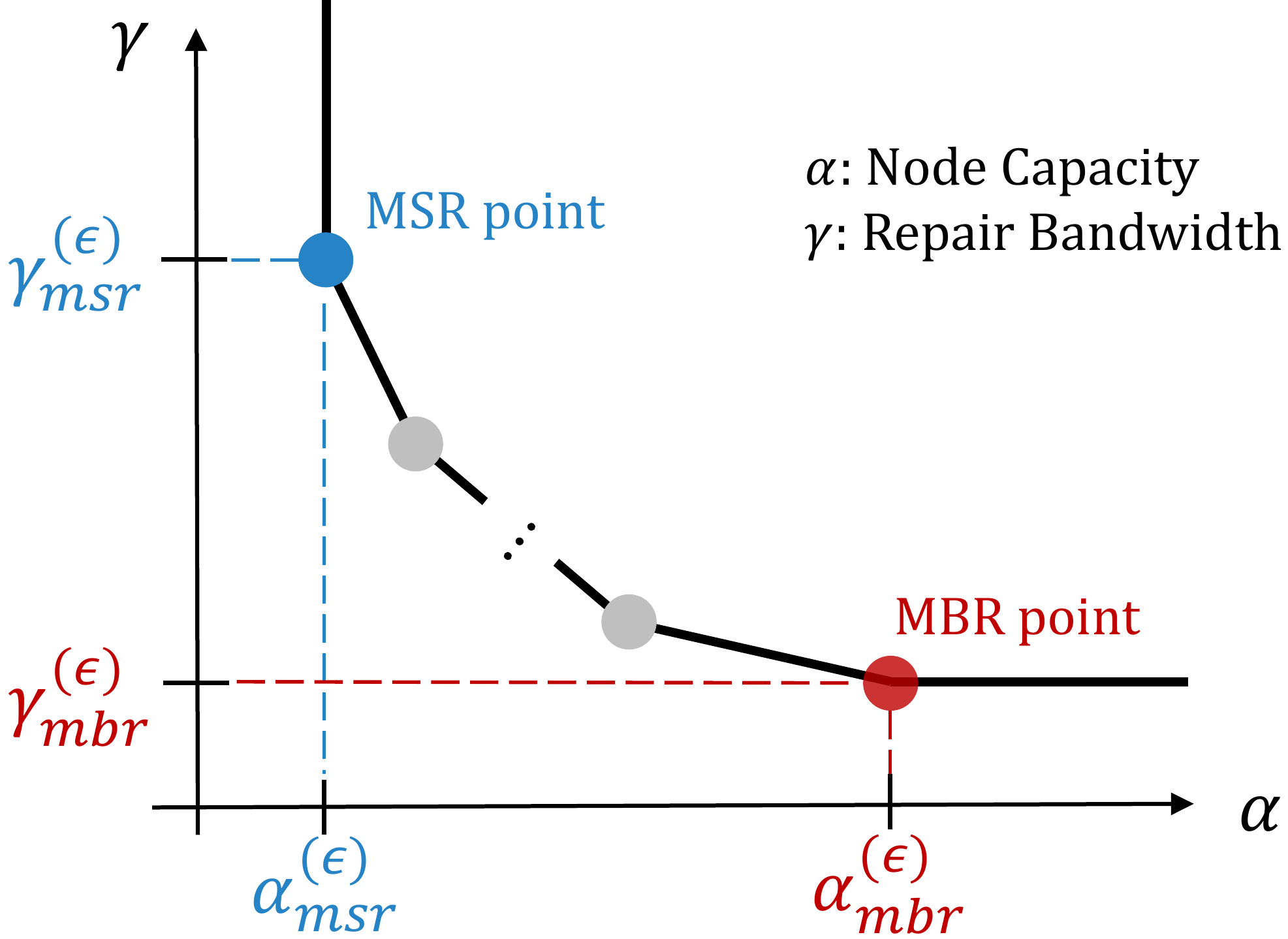}
	\caption{Set of feasible $(\alpha, \gamma)$ points for arbitrary given $\epsilon=\beta_c/\beta_I$. The feasible points achieve capacity $\mathcal{C}(\alpha, \gamma) = \mathcal{M}$ by using minimum resources. We focus on minimum-bandwidth-regenerating (MBR) point and minimum-storage-regenerating (MSR) point.}
	\label{Fig:MBR_MSR_points}
	\vspace{-0.15in}
\end{figure}

\subsection{Target Problem: Constructing MBR and MSR Codes for Clustered Distributed Storage}
We consider clustered distributed storage systems specified by parameters ($n,k,L,\alpha,\beta_I,\beta_c$). Recall that the maximum file size that can be reliably stored in this system is capacity $\mathcal{C}(\alpha, \beta_I, \beta_c)$ in \eqref{Eqn:Capacity of clustered DSS_rev}. For a given $\epsilon=\beta_c / \beta_I$, we denote the capacity as $\mathcal{C}(\alpha, \gamma)$ using two important system parameters: \textit{node storage capacity} $\alpha$ and \textit{repair bandwidth} $\gamma$ in \eqref{Eqn:gamma}.
Consider the target of reliably storing a file with size $\mathcal{M}$. The feasible 
$(\alpha,\gamma)$ points which satisfy $\mathcal{C}(\alpha,\gamma) = \mathcal{M}$
are obtained in Corollary 1 of \cite{sohn2018capacity}, and are illustrated in Fig. \ref{Fig:MBR_MSR_points} which show a tradeoff relationship.
In this figure, we focus on two important points $-$ MBR and MSR points $-$ which achieve capacity by using minimum system resources. Among the points having the minimum repair bandwidth $\gamma$, one with the smallest $\alpha$ is called the minimum-bandwidth-regenerating (MBR) point. The explicit regenerating coding schemes which achieve the MBR point is called MBR codes.
Also, among the points with the minimum node capacity $\alpha$, one having the smallest $\gamma$ is called the minimum-storage-regenerating (MSR) point. The regenerating codes which achieve the MSR point is called the MSR codes.
We denote the resource pair for an MBR code as $(\alpha, \gamma) = (\alpha_{\text{mbr}}^{(\epsilon)}, \gamma_{\text{mbr}}^{(\epsilon)})$. Similarly, an MSR code has the resource pair $(\alpha, \gamma) = (\alpha_{\text{msr}}^{(\epsilon)}, \gamma_{\text{msr}}^{(\epsilon)})$.
Note that the explicit form of the resource pairs are given in Section \ref{Section:ResourcePair}.

We aim at constructing MBR and MSR codes which achieve storage capacity $\mathcal{C}(\alpha, \gamma) = \mathcal{M}$.
Under the clustered system model, the necessary and sufficient condition for a code being MBR (or MSR) is given as below.
The validity of the suggested codes can be confirmed from Condition \ref{Condition:code_general}.
\begin{condition}[Necessary and sufficient condition for valid MBR/MSR codes]\label{Condition:code_general}
	Consider a clustered distributed storage system with parameters  
	($n,k,L,\alpha,\beta_I,\beta_c$), and denote $\epsilon=\beta_c/\beta_I$. Under this setting, a code is MBR if and only if it satisfies the following properties for $(\alpha, \gamma) = (\alpha_{\text{mbr}}^{(\epsilon)}, \gamma_{\text{mbr}}^{(\epsilon)})$:
	\begin{itemize}
		\item Each node contains $\alpha$ symbols.
		\item \textbf{(Exact regeneration)} Suppose a node fails. Then, each survived node within the cluster containing the failed node transmits $\beta_I$ symbols and each node residing in another cluster transmits $\beta_c =\epsilon \beta_I$ symbols. In total, $\gamma = (n_I-1)\beta_I + (n-n_I) \beta_c $ symbols are transmitted in the repair process.
		\item \textbf{(Data reconstruction)} Contacting arbitrary $k$ out of $n$ nodes suffices to recover the original file of size $\mathcal{M} = \mathcal{C} (\alpha, \gamma)$.
	\end{itemize}
Similarly, a code is MSR if and only if the three properties above hold for $(\alpha, \gamma) = (\alpha_{\text{msr}}^{(\epsilon)}, \gamma_{\text{msr}}^{(\epsilon)})$.
\end{condition}

In the first part of this paper, we design explicit MBR codes for clustered distributed storage with arbitrary system parameters of $(n,k,L,\alpha, \beta_I, \beta_c)$. 
The explicit code construction depends on the $\epsilon=\beta_c/\beta_I$ ratio. In Section \ref{Sec:CodeDesign_zero_epsilon}, an MBR code is constructed when $\epsilon=0$, or equivalently $\beta_c = 0$. Moreover, the MBR code construction for $0 < \epsilon \leq 1$ (i.e., $\beta_c \neq 0$) is provided in Section \ref{Sec:CodeDesign_nonzero_epsilon}.
In the second part of this paper, we design MSR codes for clustered distributed storage. According to Theorem 3 of \cite{sohn2018capacity}, node storage capacity of an MSR code satisfies
\begin{align}
\alpha_\text{msr}^{(\epsilon)} &= \mathcal{M}/k \quad \quad \text{if  } \frac{1}{n-k} \leq \epsilon \leq 1, \label{Eqn:alpha_MSR_large_epsilon_val} \\
\alpha_\text{msr}^{(\epsilon)} &> \mathcal{M}/k \quad \quad \text{if  } 0 \leq \epsilon < \frac{1}{n-k}. \label{Eqn:alpha_MSR_small_epsilon_val}
\end{align}
Recall that $\alpha = \mathcal{M}/k$ is the minimum storage overhead of each node in order to recover file size $\mathcal{M}$ by contacting $k$ nodes, 
as stated in \cite{dimakis2010network}. Thus, the minimum node storage overhead $\alpha=\mathcal{M}/k$ is achievable only for $\epsilon \in [\frac{1}{n-k}, 1]$.
Motivated by this fact, we focus on two cases: Section \ref{Section:MSR_epsilon_0} provides the explicit MSR code construction when $\epsilon=0$, and MSR codes for the $\epsilon \in [\frac{1}{n-k}, 1]$ case is given in Section \ref{Section:MSR_epsilon_positive}.

\subsection{Resource Pairs $(\alpha, \gamma)$ for MBR and MSR Codes}\label{Section:ResourcePair}

Here we recall explicit expressions for $(\alpha, \gamma)$ resource pairs of MBR and MSR codes for clustered distributed storage, which are stated in \cite{sohn2018capacity}. 
First, we define notations of
\begin{align}
q &= \left\lfloor\dfrac{k}{n_I}\right\rfloor, \label{Eqn:quotient} \\
r &= k \Mod{n_I} = k-qn_I, \label{Eqn:remainder}
\end{align}
which represent the quotient and the remainder of $k/n_I$.

\begin{prop}[A modified version of Corollary 2 of \cite{sohn2018capacity}]
	Consider a clustered distributed storage system \cite{sohn2018capacity}, which aims at storing file of size $\mathcal{M}$. For a given $\epsilon=\beta_c/\beta_I$ satisfying $0 \leq \epsilon \leq 1$, the resource pair $(\alpha, \gamma)$ for an MBR code is 
	\begin{equation}\label{Eqn:MBR_resource_pair}
	(\alpha_{\text{mbr}}^{(\epsilon)}, \gamma_{\text{mbr}}^{(\epsilon)}) = 
	(\mathcal{M}/s_0^{(\epsilon)}, \mathcal{M}/s_0^{(\epsilon)})
	\end{equation}
	where
	\begin{align}
	s_0^{(\epsilon)} &= \frac{\sum_{i=1}^k \big\{(n_I - h_i) + \epsilon (n-n_I-i+h_i) \big\}}{n_I - 1 + \epsilon (n-n_I)}, \\
	h_i &= \min \{t \in [n_I]: \sum_{m=1}^{t} g_m \geq i \}, \label{Eqn:h_i}
	\end{align}
	and $g_m$ is defined in \eqref{Eqn:g_m}. Moreover, the resource pair $(\alpha, \gamma)$ for an MSR code is expressed as 
	\begin{equation}\label{Eqn:MSR_resource_pair}
	(\alpha_{\text{msr}}^{(\epsilon)}, \gamma_{\text{msr}}^{(\epsilon)}) =
	\begin{cases}
	\left(\frac{\mathcal{M}}{k-q}, \frac{\mathcal{M}}{k-q} (n_I-1)\right), & \epsilon=0 \\
	\left(\frac{\mathcal{M}}{k}, \frac{\mathcal{M}}{k} \cdot \frac{n-n_I+(n_I-1)/\epsilon}{n-k}\right), & \frac{1}{n-k} \leq \epsilon \leq 1
	\end{cases}	
	\end{equation}
	where $q$ is in \eqref{Eqn:quotient}.
\end{prop}
\begin{proof}
The resource pair of an MBR code can be obtained directly from Corollary 2 of \cite{sohn2018capacity}. Regarding an MSR code, the proofs for $\epsilon=0$ and $\frac{1}{n-k} \leq \epsilon \leq 1$ cases are given in Appendices \ref{Section:proof_of_prop_param_small_epsilon} and \ref{Section:proof_of_prop_param_large_epsilon}, respectively.
\end{proof}
Note that the storage node capacity is equal to the repair bandwidth, i.e., 
\begin{equation} \label{Eqn:alpha_gamma_MBR}
\alpha_{\text{mbr}}^{(\epsilon)} = \gamma_{\text{mbr}}^{(\epsilon)},
\end{equation}
for MBR codes with $0 \leq \epsilon \leq 1$.

\subsection{Notations}
Throughout the paper, we use some useful additional notations. For a positive integer $n$, we denote $\{1,2,\cdots, n\}$ as $[n]$. 
For positive integers $a$ and $b$, we use the notation $a \divides b$ if $a$ divides $b$. Similarly, we write $a \notdivides b$ if $a$ cannot divide $b$.
Moreover, a clustered distributed storage system having parameters of $n,k,L$ is denoted as an $[n,k,L]-$clustered DSS.
Note that this clustered system can be expressed as a two-dimensional representation, as in Fig. \ref{Fig:2dim_representation}. In this structure, we denote the $j^{th}$ storage node in the $l^{th}$ cluster as $N{(l,j)}$. 
A vector is denoted as $\mathbf{v}$ using a bold-faced lower case letter. 
For positive integers $m$ and $n$, the set $\{y_m, y_{m+1}, \cdots, y_n\}$ is represented as $\{y_i\}_{i=m}^{n}$.
The binomial coefficient $\frac{n!}{k! (n-k)!}$ is written as $\binom{n}{k}$.

Finally, we recall definitions on the locally repairable codes (LRCs) in \cite{papailiopoulos2014locally, tamo2016optimal}.  As defined in \cite{tamo2016optimal}, an $(n,k,r)-$LRC represents a code of length  $n$, which is encoded from $k$ information symbols. Every coded symbol of the $(n,k,r)-$LRC can be regenerated by accessing at most $r$ other symbols. As defined in \cite{papailiopoulos2014locally}, an $(n,r,d,\mathcal{M},\alpha)-$LRC takes a file of size $\mathcal{M}$ and encodes it into $n$ coded symbols, where each symbol is composed of $\alpha$ bits. Moreover, any coded symbol can be regenerated by contacting at most $r$ other symbols, and the code has the minimum distance of $d$.

\begin{figure}[!t]
	\centering
	\includegraphics[height=25mm]{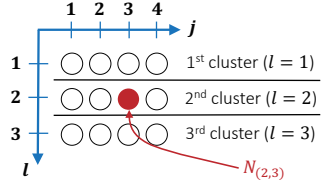}
	\caption{Two-dimensional representation of clustered distributed storage ($n=12, L=3, n_I = n/L = 4$)}
	\label{Fig:2dim_representation}
	\vspace{-0.15in}
\end{figure}

\section{MBR Code Design for $\epsilon = 0$}\label{Sec:CodeDesign_zero_epsilon}

We construct an MBR code and confirm that the proposed code satisfies Condition \ref{Condition:code_general} and achieves capacity of clustered distributed storage. This section considers scenarios with zero cross-cluster repair bandwidth, i.e., $\beta_c = 0$ or $\epsilon=0$.

\subsection{Parameter Setting for MBR Code with $\epsilon=0$}

Without a loss of generality, we set $\beta_I = 1$. Then, from \eqref{Eqn:alpha_gamma_MBR} and \eqref{Eqn:gamma}, the resource pair $(\alpha, \gamma) = (\alpha_{\text{mbr}}^{(0)}, \gamma_{\text{mbr}}^{(0)})$ is given as
\begin{equation}\label{Alpha_MBR_GammaC0}
\alpha_{\text{mbr}}^{(0)} = \gamma_{\text{mbr}}^{(0)} = (n_I-1)\beta_I + (n-n_I)\beta_c = n_I - 1.
\end{equation}
Moreover, combining (\ref{Eqn:MBR_resource_pair}) and (\ref{Alpha_MBR_GammaC0}), the maximum reliably storable file size $\mathcal{M}$ is expressed as
\begin{equation}\label{Eqn:Capacity for gamma_c = 0}
\mathcal{M}=s_0^{(0)} \gamma_{\text{mbr}}^{(0)} = s_0^{(0)} (n_I - 1) = \sum_{i=1}^k (n_I-  h_i).
\end{equation}

\subsection{Code Construction}\label{Section:Code design for zero gammac}

We now propose an explicit coding scheme with parameters $\alpha = \gamma =  n_I - 1, \beta_I =1, \beta_c = 0$, satisfying the following:
\begin{itemize}
	\item A failed node is repaired within the cluster, by receiving $\beta_I = 1$ symbol from each node in the same cluster.
	\item Contacting any $k$ out of $n$ nodes can recover file size $\mathcal{M}$.
\end{itemize}
The suggested coding scheme is based on the repair-by-transfer (RBT) scheme devised in \cite{rashmi2009explicit, shah2012TIT}. Compared to these existing works on non-clustered (homogeneous) DSSs, the present paper applies the RBT scheme by reflecting the clustered nature of DSSs. 
Before specifying the code construction rule, we define a matrix $V_t$ as follows.

\begin{definition}
Consider a fully connected graph $G_t$ with $t$ vertices. Then, $V_t$ is defined as the incidence matrix of $G_t$, which is a $t \times {t \choose 2}$ matrix given by:
\begin{equation}
V_t(j,i)=
\begin{cases}
1, & \text{if } i^{th} \text{ edge is connected to } j^{th} \text{ node} \\
0, & \text{otherwise}.
\end{cases}
\end{equation}
\end{definition}
Note that Fig. \ref{Fig:incidence_matrix} gives an example of graph $G_t$ and its incidence matrix $V_t$ for $t=4$. Using this definition, we provide an explicit code construction rule suitable for the $\epsilon=0$ case in Algorithm \ref{Algo:MBR_code_zero}, under the setting of arbitrary $n,k,L$. 
For given $\mathcal{M}$ source symbols $\mathbf{s}= [s_1, \cdots, s_{\mathcal{M}}]$, the algorithm specifies 1) the encoding rule and 2) the rule for distributing coded symbols to $n$ nodes $\{N(l,j)\}_{l \in [L], j \in [n_I]}$ in $L$ clusters.

\begin{algorithm}[!t]
	\caption{MBR code construction for $\epsilon =0$}
	\label{Algo:MBR_code_zero}
	\begin{algorithmic}
		\REQUIRE System parameters $n, k, L$ and $\mathcal{M}$ in \eqref{Eqn:Capacity for gamma_c = 0},
		\\ \hspace{6.5mm} Source symbol vector $\mathbf{s}= [s_1, \cdots, s_{\mathcal{M}}]^T$. 
		\ENSURE Symbols stored on nodes $\{N(l,j)\}_{l \in [L], j \in [n_I]}$
		\STATE \textbf{Step 1.} Generate encoded symbols $\{c_1, \cdots, c_{\theta}\}$: \STATE \hspace{0mm}Apply a $(\theta, \mathcal{M})-$MDS code to source symbol vector $\mathbf{s}$, resulting in $\mathbf{c}=[c_1, \cdots, c_{\theta}]^T$ . Here, we have
		\begin{align}\label{Eqn:theta_gammac0}
		 \theta = {n_I\choose 2} L .
		 \end{align}
		\STATE \textbf{Step 2.} Distribute encoded symbols to nodes under the following rule: 
		\STATE \begin{itemize}
			\item Node $N(l,j)$ stores symbol $c_{(l-1){n_I \choose 2}+i}$ if and only if $V_{n_I}(j,i) = 1$. Here, the ranges of parameters are $l\in [L],  j\in [n_I], $ and $ i\in \left[ {n_I \choose 2} \right].$
		\end{itemize}
	\end{algorithmic}
\end{algorithm}

The suggested coding scheme has the following properties, which are useful for proving Theorem \ref{Thm:CodeZero}.

\begin{lemma}\label{Prop:MBR for zero gammac}
Suppose the code in Algorithm \ref{Algo:MBR_code_zero} is applied to an $[n,k,L]-$clustered DSS with $\epsilon=0$. Then, the system satisfies all of the following:
	\begin{enumerate}[label=(\alph*)]
		\item Each coded symbol $c_i$ is stored in exactly two different storage nodes. \label{Prop:MBR for zero gammac_first}
		\item Nodes in different clusters do not share any coded symbols. \label{Prop:MBR for zero gammac_second}
		\item Nodes in the same cluster share exactly one coded symbol.\label{Prop:MBR for zero gammac_third}
		\item Each node contains $\alpha = n_I - 1$ coded symbols.\label{Prop:MBR for zero gammac_fourth}
	\end{enumerate}
\end{lemma}
\begin{proof}
	See Appendix \ref{Proof:Properties of suggested MBR codes gammac=0 case}.
\end{proof}

\begin{figure}[!t]
	\centering
	\includegraphics[height=35mm]{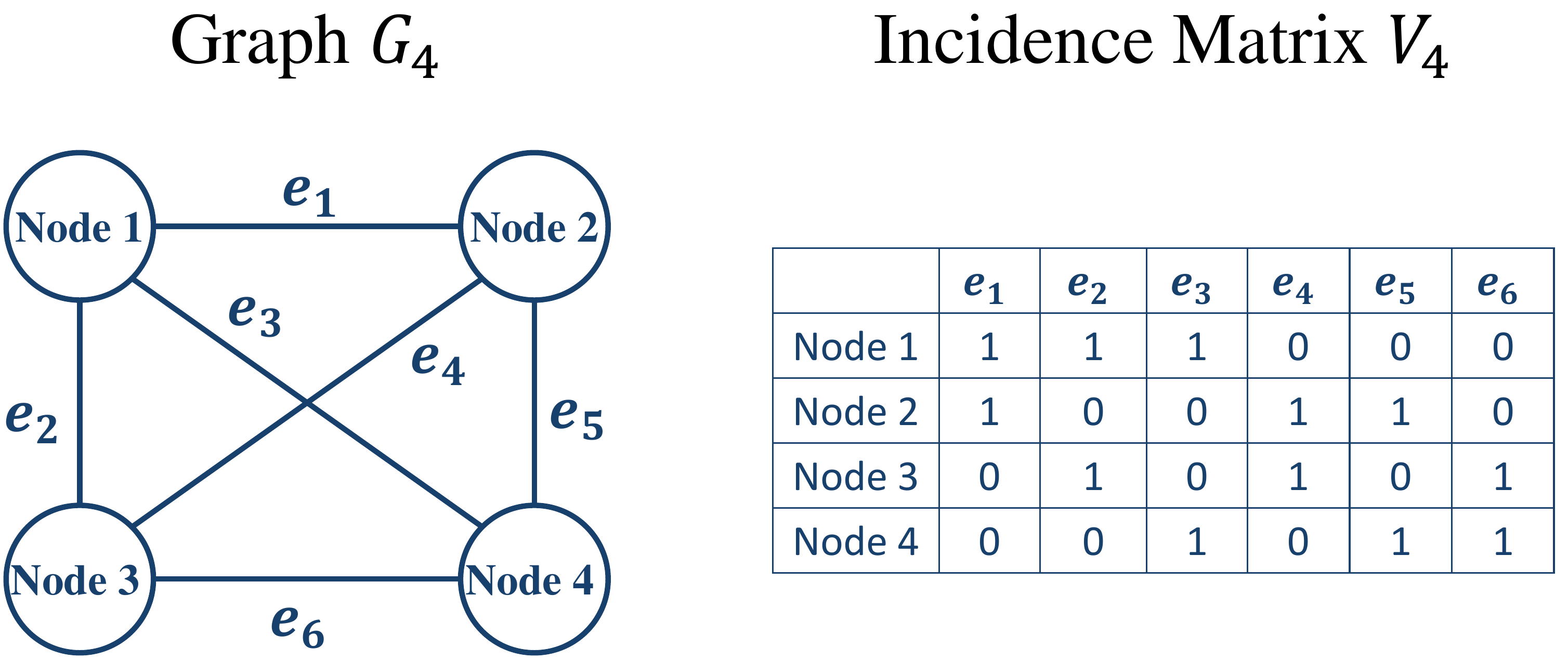}
	\caption{Incidence matrix $V_{t}$ of a fully connected graph $G_t$ with $t =4$.}
	\label{Fig:incidence_matrix}
	\vspace{-0.1in}
\end{figure} 

Using this lemma, we show that Algorithm \ref{Algo:MBR_code_zero} provides a valid MBR code for any $n,k,L$ setting with $\epsilon=0$. 

\begin{theorem}\label{Thm:CodeZero}	
The code suggested in Algorithm \ref{Algo:MBR_code_zero} is an MBR code for any $[n,k,L]-$clustered DSS with $\epsilon=0$. In other words, it satisfies all requirements stated in Condition \ref{Condition:code_general}:
\begin{itemize}
	\item 
	Each node contains $\alpha_{\text{mbr}}^{(0)}=n_I-1$ coded symbols.
	\item \textbf{(Exact regeneration)} When a node fails, it can be exactly regenerated by using the intra-cluster repair bandwidth of $\beta_I=1$ and the cross-cluster repair bandwidth of $\beta_c=0$. Thus, it has the total repair bandwidth of $\gamma_{\text{mbr}}^{(0)}=n_I-1$ coded symbols.
	\item \textbf{(Data reconstruction)} Contacting any $k$ out of $n$ nodes can retrieve vector $\mathbf{s}$ consisting of $\mathcal{M}$ source symbols.
\end{itemize}
\end{theorem}
\begin{proof}
See Appendix \ref{Proof:Thm:CodeZero}.
\end{proof}

\subsection{Example}

\begin{figure}[!t]
	\centering
	\includegraphics[height=50mm]{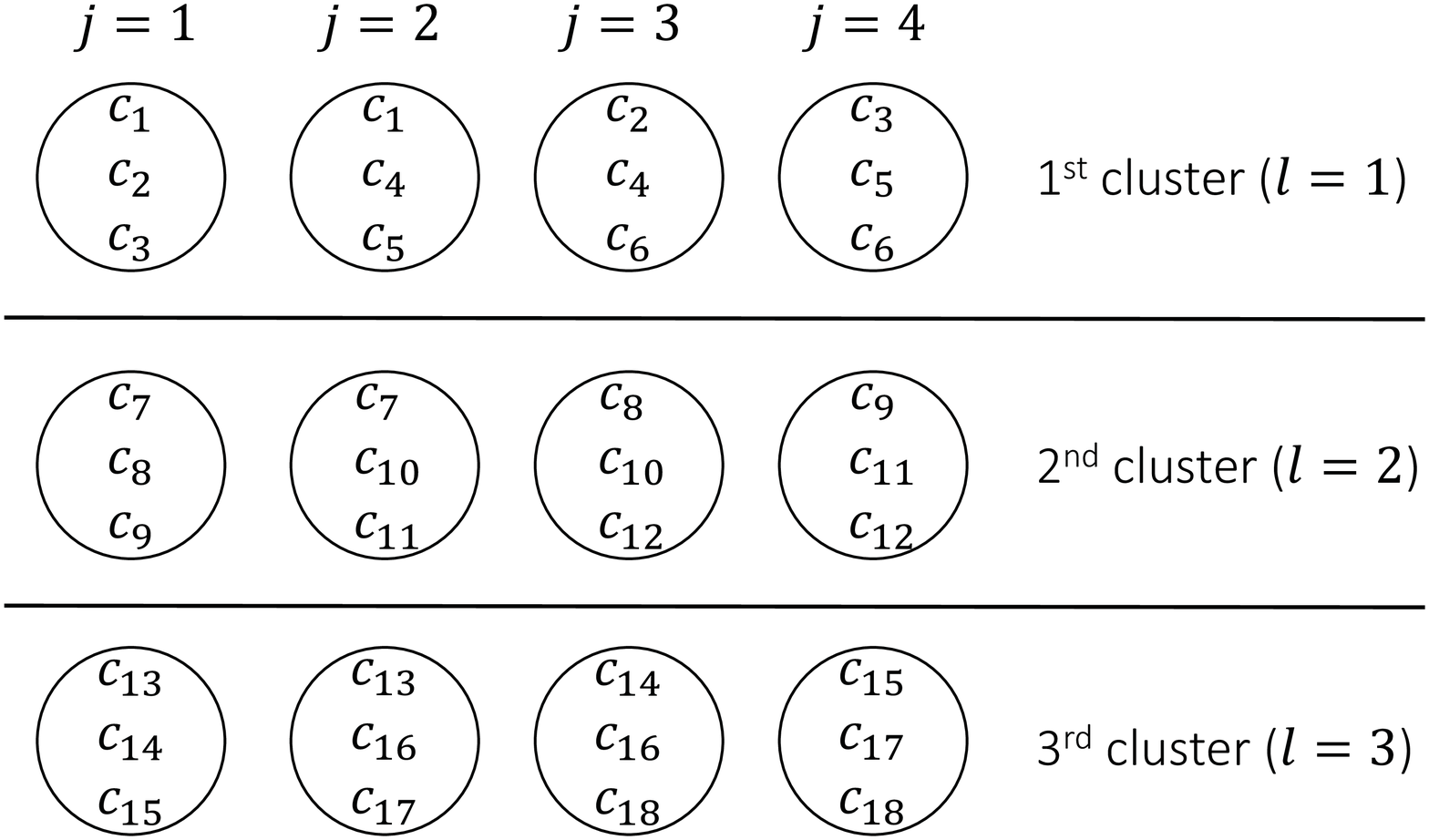}
	\caption{MBR coding scheme example for $n=12, L=3, \epsilon=0$}
	\label{Fig:MBR_code_example_gammac0}
	\vspace{-0.1in}
\end{figure} 
Consider a clustered distributed storage system with $n=12, k=6$ and $L=3$, as in Fig. \ref{Fig:MBR_code_example_gammac0}. This implies $n_I = \frac{n}{L}= 4,  \theta=18$ and $\mathcal{M}=11$ from \eqref{Eqn:Capacity for gamma_c = 0} and \eqref{Eqn:theta_gammac0}. 
Suppose the code in Algorithm \ref{Algo:MBR_code_zero} is applied to the system, for given 
$\mathcal{M}=11$ source symbols $\mathbf{s}=[s_1, s_2, \cdots, s_{11}]^T$.
As a first step, $[\theta, \mathcal{M}]=[18,11]-$MDS code is applied for encoding, which results in $\theta = 18$ coded symbols $\mathbf{c} = [c_1, c_2, \cdots, c_{18}]^T$.
Next, the coded symbols are allocated into different nodes as follows: node $N(l,j)$ stores $c_{6(l-1)+i}$ if and only if $V_4(j,i) = 1$. Here, the ranges of parameters are $l \in [3], j \in [4]$ and $i \in [6]$. Using matrix $V_4$ in Fig. \ref{Fig:incidence_matrix}, the coded symbols are distributed as in  
Fig. \ref{Fig:MBR_code_example_gammac0}. Each circle represents a storage node, which stores $\alpha = 3$ coded symbols. As an example of $(l,j) = (2,3)$, node $N{(2,3)}$ stores three coded symbols: $c_8, c_{10}$ and $c_{12}$.

\textit{Exact Regeneration:} Suppose node $N(2,3)$ fails. Then, node $N(2,1)$ gives $c_8$, node $N(2,2)$ gives $c_{10}$, and node $N(2,4)$ gives $c_{12}$. These three symbols are stored in the new storage node for replacing $N(2,3)$. This regeneration process satisfies $\beta_I = 1$ and $\beta_c = 0$. The same argument holds for arbitrary single-node failure events.

\textit{Data Reconstruction:}
Suppose that the data collector connects to $k=6$ nodes, say, $N(1,1)$, $N(1,2)$, $N(1,3)$, $N(1,4)$, $N(2,1),$ and $N(2,2)$. It is easy to see that this process allows the retrieval of the eleven coded symbols $\{c_i\}_{i=1}^{11}$. By using the MDS property of the $[\theta,\mathcal{M}]=[18,11]$ code used in Algorithm \ref{Algo:MBR_code_zero}, we can easily obtain $\{s_i\}_{i=1}^{11}$ from the retrieved coded symbols $\{c_i\}_{i=1}^{11}$. 
The same can be shown for any arbitrary choice of $k=6$ nodes.

\subsection{Comparison with existing local MBR codes}\label{Sec:comparison_with_local_MBR}

Recall that when $\epsilon=0$, no cross-cluster repair bandwidths are allowed. This scenario reminds one of a class of codes called \textit{locally repairable codes (LRCs)} \cite{papailiopoulos2014locally}, which regenerate a failed node by contacting some limited number of nodes. There have been extensive works on LRCs, with some \cite{kamath2013ISIT,kamath2014TIT} especially focusing on constructing LRCs with the MBR property. These codes are called  \textit{local MBR codes}.
Here we compare our MBR code suggested in Algorithm \ref{Algo:MBR_code_zero} with the local MBR codes previously designed by other researchers. 
In Construction 6.3 of \cite{kamath2014TIT}, the authors proposed a local MBR code based on the repair-by-transfer (RBT) scheme, similar to the code suggested in the current paper. However, the code in \cite{kamath2014TIT} requires a finite field with size  $\binom{n\alpha/2}{\mathcal{M}}=\binom{n(n_I-1)/2}{\mathcal{M}}$, while the code in the current paper requires a much smaller field size of $\theta = \binom{n_I}{2} L = n(n_I-1)/2$, when RS code is used for $[\theta, \mathcal{M}]-$MDS encoding. 
In addition, the local MBR code suggested in Construction II.1 of \cite{kamath2013ISIT}
requires a finite field with size much larger than that of the MBR code suggested in the current paper. To be specific, the required field size of the code in \cite{kamath2013ISIT} is $q^m$ for some prime number $q\geq 2$ and $m \geq L K_{\text{local}}$, where $K_{\text{local}}$ is the scalar dimension of a local code satisfying $K_{\text{local}} = \alpha (n_I - 1) - \binom{n_I-1}{2} \geq \frac{(n_I-1)^2}{2} \simeq \frac{n_I (n_I-1)}{2}$ . 
Thus, one gets $q^m \geq  q^{LK_{\text{local}}} \simeq q^{\theta}$, which is far greater than $\theta$, the field size of our MBR code.


\section{MBR Code Design for $0 < \epsilon \leq 1$ case}\label{Sec:CodeDesign_nonzero_epsilon}

We now construct an MBR code assuming that repair bandwidths across other clusters are allowed, i.e., $\beta_c \neq 0$. Using the definition of $\epsilon$, this setting corresponds to the case of $0 < \epsilon \leq 1$.

\subsection{Parameter Setting for MBR Code}\label{Section:parameter for MBR}

For a given $0 < \epsilon \leq 1 $, we consider MBR codes with resource pair $(\alpha, \gamma)=(\alpha_{\text{mbr}}^{(\epsilon)}, \gamma_{\text{mbr}}^{(\epsilon)})$, which reliably store file size $\mathcal{M} = \mathcal{C}(\alpha, \gamma)$. Let us assume that 
$\chi\coloneqq 1/\epsilon =\beta_I/\beta_c$
is a positive integer. Without losing generality, we set $\beta_I = \chi$ and $\beta_c = 1$.
From \eqref{Eqn:alpha_gamma_MBR} and \eqref{Eqn:gamma}, we have
\begin{align}\label{Eqn:resource_for_epsilon_nonzero}
\alpha_{\text{mbr}}^{(\epsilon)} = \gamma_{\text{mbr}}^{(\epsilon)} &= (n_I - 1)\chi + (n-n_I) \\
&= (n_I-1)/\epsilon + (n-n_I). \nonumber
\end{align}
Under this setting, the following proposition provides a simplified form of the capacity expression.
\begin{prop}\label{Prop:MBR_for_nonzero_beta_c}
	Consider an MBR point having a resource pair of $(\alpha, \gamma) = (\alpha_{\text{mbr}}^{(\epsilon)}, \gamma_{\text{mbr}}^{(\epsilon)})$ with $0 < \epsilon \leq 1$. When repair bandwidths are set to $\beta_c = 1$ and $\beta_I = \chi$ for some positive integer $\chi = 1/\epsilon$, the capacity of clustered distributed storage can be expressed as
	\begin{equation}\label{Eqn:Capacity for MBR}
	\mathcal{M} = \mathcal{C}(\alpha, \gamma) =  k\alpha  - \frac{1}{2}(\chi-1)(qn_I^2 + r^2 - k) - \frac{k(k-1)}{2}
	\end{equation}
	where $q$ and $r$ are defined in (\ref{Eqn:quotient}) and (\ref{Eqn:remainder}), respectively.
\end{prop}
\begin{proof}
	See Appendix \ref{Proof:Capacity for MBR}.
\end{proof}

\subsection{Code Construction}

\begin{algorithm}[!t]
	\caption{MBR code construction for $0 < \epsilon \leq 1$}
	\label{Algo:CodeNonzero}
	\begin{algorithmic}
		\REQUIRE System parameters $n, k, L, \chi = \beta_I/\beta_c$ and $\mathcal{M}$ in \eqref{Eqn:Capacity for MBR},
		\\ \hspace{6.5mm} Source symbol vector $\mathbf{s} = [s_1, \cdots, s_{\mathcal{M}}]^T$.
		\ENSURE Symbols stored on nodes $\{N(l,j)\}_{l \in [L], j \in [n/L]}$
		\STATE \textbf{Step 1.} Generate encoded symbols $\{c_1, \cdots, c_{\theta}\}$: \STATE \hspace{0mm}Apply a $(\theta, \mathcal{M})-$MDS code to source symbol vector $\mathbf{s}$, resulting in $\mathbf{c}=[c_1, \cdots, c_{\theta}]^T$. Here, we have
		\begin{equation}\label{Eqn:theta}
		\theta = (\chi-1)  {n_I\choose 2} L + {n \choose 2}.
		\end{equation}
		\STATE \textbf{Step 2.} Distribute encoded symbols to nodes under the following rules. Among $\theta$ coded symbols, the first $\binom{n}{2}$ symbols are called \textit{global} symbols, and the rest $ (\chi-1)  {n_I\choose 2} L$ symbols are called \textit{local} symbols. We have different allocation rules for global/local symbols as below.
		\STATE \begin{itemize}
			\item Node $N(l,j)$ stores a \textit{global} symbol $c_{i_1}$ if and only if $V_n (n_I(l-1) +j,i_1) =1$.
			\item For every $t \in [\chi-1]$, follow the rule: 
			Node $N(l,j)$ stores a \textit{local} symbol $c_{{n \choose 2} + (\chi l - \chi - l + t){n_I \choose 2}+i_2}$ if and only if $V_{n_I}(j,i_2)=1$. 
			\item Here, the ranges of the parameters are $l \in [L], j\in [n_I], i_1\in [ n(n-1)/2 ] $, and $i_2\in [ n_I(n_I-1)/2 ] .$
		\end{itemize}
	\end{algorithmic}
\end{algorithm}


In Algorithm \ref{Algo:CodeNonzero}, we propose an MBR code for $0 < \epsilon \leq 1$. 
In the symbol allocation rule in Step 2 of Algorithm \ref{Algo:CodeNonzero}, coded symbols are classified as \textit{global symbols} and \textit{local symbols}.
Note that \textit{global symbols} are shared either among nodes in the same cluster or among nodes in other clusters, while \textit{local symbols} are shared among nodes in the same cluster only. An example of code construction in Fig. \ref{Fig:MBR_code_example} with an explanation in Section \ref{Sec:CodeExample_nonzero} provides an insight on global/local symbols.
Now we provide properties of the code in Algorithm \ref{Algo:CodeNonzero}, which are useful for proving Theorem \ref{Thm:CodeNonzero}.

\begin{lemma}\label{Prop:MBR}
Suppose the code in Algorithm \ref{Algo:CodeNonzero} is applied to an $[n,k,L]-$clustered DSS with $0 < \epsilon \leq 1$. Then, the system satisfies the following:
	\begin{enumerate}[label=(\alph*)]
		\item Each coded symbol $c_i$ is stored in exactly two different storage nodes. \label{Prop:MBR_first}
		\item Nodes in different clusters share one coded symbol. \label{Prop:MBR_second}
		\item Nodes in the same cluster share $\chi$ coded symbols. \label{Prop:MBR_third}
		\item Each node contains $\alpha = (n_I-1)\chi + (n-n_I)$ coded symbols. \label{Prop:MBR_fourth}
	\end{enumerate}
\end{lemma}
\begin{proof}
	See Appendix \ref{Proof:Properties of suggested MBR codes}.
\end{proof}

Based on this Lemma, the code in Algorithm \ref{Algo:CodeNonzero} is shown to be an MBR code for any $n,k,L$ setting with $0 < \epsilon \leq 1$, as stated in the following Theorem.

\begin{theorem}\label{Thm:CodeNonzero}
The code suggested in Algorithm \ref{Algo:CodeNonzero} is an MBR code for any $[n,k,L]-$clustered DSS with $0 < \epsilon \leq 1$. In other words, it satisfies all requirements stated in Condition \ref{Condition:code_general}:
\begin{itemize}
	\item 
	Each node contains $\alpha_{\text{mbr}}^{(\epsilon)}=(n_I-1)/\epsilon + (n-n_I)$ coded symbols.
	\item \textbf{(Exact regeneration)} When a node fails, it can be exactly regenerated by using the intra-cluster repair bandwidth of $\beta_I=\chi$ and the cross-cluster repair bandwidth of $\beta_c=1$. Thus, it has the total repair bandwidth of $\gamma_{\text{mbr}}^{(\epsilon)}=(n_I-1)/\epsilon + (n-n_I)$ coded symbols.
	\item \textbf{(Data reconstruction)} Contacting any $k$ out of $n$ nodes can retrieve vector $\mathbf{s}$ consisting of $\mathcal{M}$ source symbols.
\end{itemize}
\end{theorem}
\begin{proof}
	See Appendix \ref{Proof:Thm:CodeNonzero}.
\end{proof}

\subsection{Example}\label{Sec:CodeExample_nonzero}

Consider the scenario where $n=6, k=3, L=2, \beta_I = 3$ and $\beta_c = 1$ (i.e., $\chi=3$ or $\epsilon=1/3$), as in Fig. \ref{Fig:MBR_code_example}. This implies $n_I=n/L=3$ and $\alpha=(n_I-1)\chi + (n-n_I) = 9$. Moreover, $\mathcal{M}=18$ and $\theta=27$ from \eqref{Eqn:Capacity for MBR}, \eqref{Eqn:theta}. First, we apply a $[\theta,\mathcal{M}]=[27,18]-$MDS code to the original source symbol $\bm{s}=[s_1, s_2, \cdots, s_{18}]^T$, which results in $\bm{c}=[c_1, c_2, \cdots, c_{27}]^T$.
Then, the coded symbols are distributed as follows: node $N(l,j)$ stores 
\begin{itemize}
	\item $c_{i_1}$ if and only if $V_6(3(l-1)+j,i_1)=1$.
	\item $c_{15 + 3(2l+t-3)+i_2}$  
	if and only if $V_{3}(j,i_2)=1$, for $t=1,2$
\end{itemize}
Here, $l \in [2] , j \in [3], i_1 \in [15] , i_2 \in [3] $. Using the following $V_6$ and $V_3$ matrices, 
\begin{equation*}
V_6 = \begin{bmatrix}
1 & 1 & 1 & 1 & 1 & 0 & 0 & 0 & 0 & 0 & 0 & 0 & 0 & 0 & 0  \\         
1 & 0 & 0 & 0 & 0 & 1 & 1 & 1 & 1 & 0 & 0 & 0 & 0 & 0 & 0 	\\	
0 & 1 & 0 & 0 & 0 & 1 & 0 & 0 & 0 & 1 & 1 & 1 & 0 & 0 & 0	\\	
0 & 0 & 1 & 0 & 0 & 0 & 1 & 0 & 0 & 1 & 0 & 0 & 1 & 1 & 0     \\    
0 & 0 & 0 & 1 & 0 & 0 & 0 & 1 & 0 & 0 & 1 & 0 & 1 & 0 & 1       \\   
0 & 0 & 0 & 0 & 1 & 0 & 0 & 0 & 1 & 0 & 0 & 1 & 0 & 1 & 1         
\end{bmatrix} \\
\end{equation*}
\begin{equation*}
V_3 = \begin{bmatrix}
1 & 1 & 0 \\
1 & 0 & 1 \\
0 & 1 & 1
\end{bmatrix}
\end{equation*}
the coded symbols are distributed as in  
Fig. \ref{Fig:MBR_code_example}. 
Each rectangle represents a storage node, which stores $\alpha = 9$ coded symbols. The numbers in a rectangle represent the indices of the coded symbols contained in the node. For example, when $(l,j) = (1,2)$, node $N(l,j) = N(1,2)$ stores $\{c_i\}$ for $i=1,6,7,8,9,16,18,19$, and $21$. 
Here, note that each of the $\{c_i\}_{i=1}^{15}$ \textit{global symbols} is shared either among the nodes within the same cluster or among the nodes in different clusters. Moreover, $\{c_i\}_{i=16}^{27}$ are \textit{local symbols}: $\{c_i\}_{i=16}^{21}$ are shared among the nodes in the $1^{st}$ cluster only, while $\{c_i\}_{i=22}^{27}$ are shared among the nodes in the $2^{nd}$ cluster only.

\begin{figure}[!t]
	\centering
	\includegraphics[width=75mm]{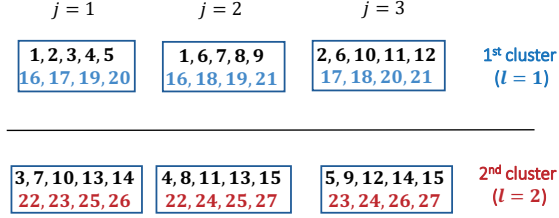}
	\caption{MBR coding scheme example for $n=6, L=2, \chi=3$. Each rectangle represents a node containing $\alpha=9$ coded symbols. To avoid cluttered notation, we just wrote the indices of the coded symbols in each rectangle.}
	\label{Fig:MBR_code_example}
	\vspace{-0.1in}
\end{figure} 

\textit{Exact Regeneration:} Suppose node $N(1,2)$ is broken. Then, $N(1,1)$ provides $c_1, c_{16}, c_{19}$, while $N(1,3)$ gives $c_6, c_{18}, c_{21}$. The nodes $N(2,1), N(2,2), N(2,3)$ transmit $c_7, c_8, c_9$, respectively. This completely regenerates the failed node  $N(1,2)$. Any single node failure can be repaired in a similar way. 

\textit{Data Reconstruction:} Consider an arbitrary contact of $k=3$ nodes, say, three nodes in the $1^{st}$ cluster in Fig. \ref{Fig:MBR_code_example}. Note that this contact can retrieve $18$ coded symbols $c_i$ for $i \in \{1, 2, \cdots, 12\} \cup \{16, 17, \cdots, 21\}$. By using the MDS property of the $[\theta, \mathcal{M}] = [27,18]$ code used in Algorithm \ref{Algo:CodeNonzero}, we obtain $\{s_i\}_{i=1}^{18}$ from the retrieved $18$ coded symbols. The same can be shown for any arbitrary contact of $k=3$ nodes.

\section{MSR Code Design for $\epsilon = 0$}\label{Section:MSR_epsilon_0}

In this section, we propose MSR codes for $\epsilon=0$, \textit{i.e.}, $\beta_c = 0$. Under this setting, no cross-cluster communication is allowed in the node repair process.
First, the system parameters for the MSR point are examined. Second, two types of locally repairable codes (LRCs) suggested in \cite{papailiopoulos2014locally, tamo2016optimal} are proven to achieve the MSR point, under the settings of $n_I \divides k$ and $n_I \notdivides k$, respectively. 

\subsection{Parameter Setting for the MSR Point}

We consider the MSR point $(\alpha, \gamma)=(\alpha_\text{msr}^{(\epsilon)}, \gamma_\text{msr}^{(\epsilon)})$ which can reliably store $\mathcal{M} = \mathcal{C}(\alpha, \gamma)$.
The following property specifies the system parameters for the $\epsilon=0$ case.

\begin{prop}\label{Prop:parameter_for_small_epsilon}
	Consider an $[n,k,L]$ clustered DSS to reliably store file size $\mathcal{M}$. The MSR point $(\alpha, \gamma) = (\alpha_\text{msr}^{(0)}, \gamma_\text{msr}^{(0)})$ for $\epsilon=0$ is
	\begin{equation}\label{Eqn:parameters_for_small_epsilon}
	(\alpha_\text{msr}^{(0)}, \gamma_\text{msr}^{(0)}) = \left(\frac{\mathcal{M}}{k-q}, \frac{\mathcal{M}}{k-q} (n_I-1)\right),
	\end{equation}
	where $q$ is defined in (\ref{Eqn:quotient}).
	This point satisfies $\alpha=\beta_I$.
\end{prop}
\begin{proof}
	See Appendix \ref{Section:proof_of_prop_param_small_epsilon}.
\end{proof}

\subsection{Code Construction for $n_I \divides k$}

We now examine how to construct an MSR code when $n_I$ divides $k$, i.e., $n_I \divides k$ holds. 
The following theorem shows that a locally repairable code constructed in \cite{papailiopoulos2014locally} with locality $r=n_I-1$ is a valid MSR code for $n_I \divides k$.

\begin{theorem}[MSR Code Construction for $\epsilon=0, n_I \divides k$]\label{Thm:LRC1_achieves_MSR}
	Let $\mathds{C}$ be the $(n,r,d,\mathcal{M},\alpha)-$LRC constructed in \cite{papailiopoulos2014locally} for locality $r=n_I-1$. 
	Consider allocating coded symbols of $\mathds{C}$ in an $[n,k,L]-$clustered DSS, where $r+1=n_I$ nodes within the same repair group of $ \mathds{C}$ are located in the same cluster.	
	Then, the code $\mathds{C}$ is an MSR code for the $[n,k,L]-$clustered DSS under the conditions of $\epsilon=0$ and $n_I \divides k$.
\end{theorem}
\begin{proof}
	See Appendix \ref{Section:proof_of_LRC1_achieves_MSR}. 
\end{proof}

Fig. \ref{Fig:MSR_for_epsilon0_divisible} illustrates an example of the MSR code for the $\epsilon=0$ and $n_I \divides k$ case, which is constructed using the LRC in \cite{papailiopoulos2014locally}. In the $[n,k,L]=[6,3,2]-$clustered DSS scenario, the parameters are set to
\begin{align*}
\alpha &= n_I = n/L = 3, \\
\mathcal{M} &= (k-q) \alpha = (k-\floor{k/n_I}) \alpha = 6.
\end{align*}
Thus, each storage node contains $\alpha=3$ symbols, while the $[n,k,L]$ clustered DSS aims to reliably store a file of size $\mathcal{M}=6$.
This code has two properties, 1) \textit{exact regeneration} and 2) \textit{data reconstruction}:
\begin{enumerate}
	\item Any failed node can be exactly regenerated by contacting $n_I-1 = 2$ nodes in the same cluster,
	\item Contacting any $k=3$ nodes can recover the original file $\{x_i^{(j)}: i \in [3], j \in [2] \}$ of size $\mathcal{M}=6$.
\end{enumerate}

The first property is obtained from the fact that $y_i^{(1)}, y_i^{(2)}$ and $s_i = y_i^{(1)}+ y_i^{(2)}$ form a $(3,2)$ MDS code for $i \in [6]$. The second property is obtained as follows. For contacting  arbitrary $k=3$ nodes, three distinct coded symbols $\{y_{i_1}^{(1)}, y_{i_2}^{(1)}, y_{i_3}^{(1)}\}$ having superscript one and three distinct coded symbols $\{y_{j_1}^{(2)}, y_{j_2}^{(2)}, y_{j_3}^{(2)}\}$ having superscript two can be obtained for some 
$i_1, i_2, i_3 \in [6]$ and 
$j_1, j_2, j_3 \in [6]$. 
From Fig. \ref{Fig:MDS_Precoding}, the information $\{y_{i_1}^{(1)}, y_{i_2}^{(1)}, y_{i_3}^{(1)}\}$ suffice to recover $x_1^{(1)}, x_2^{(1)}, x_3^{(1)}$. Similarly, the information $\{y_{j_1}^{(2)}, y_{j_2}^{(2)}, y_{j_3}^{(2)}\}$ suffice to recover $x_1^{(2)}, x_2^{(2)}, x_3^{(2)}$. This completes the proof for the second property.  
Note that this coding scheme is already suggested by the authors of \cite{papailiopoulos2014locally}, while the present paper proves that this code also achieves the MSR point of the $[n,k,L]-$clustered DSS, in the case of $\epsilon=0$ and $n_I \divides k$.

\begin{figure}
	\centering
	\subfloat[][MDS precoding]{\includegraphics[width=80mm ]{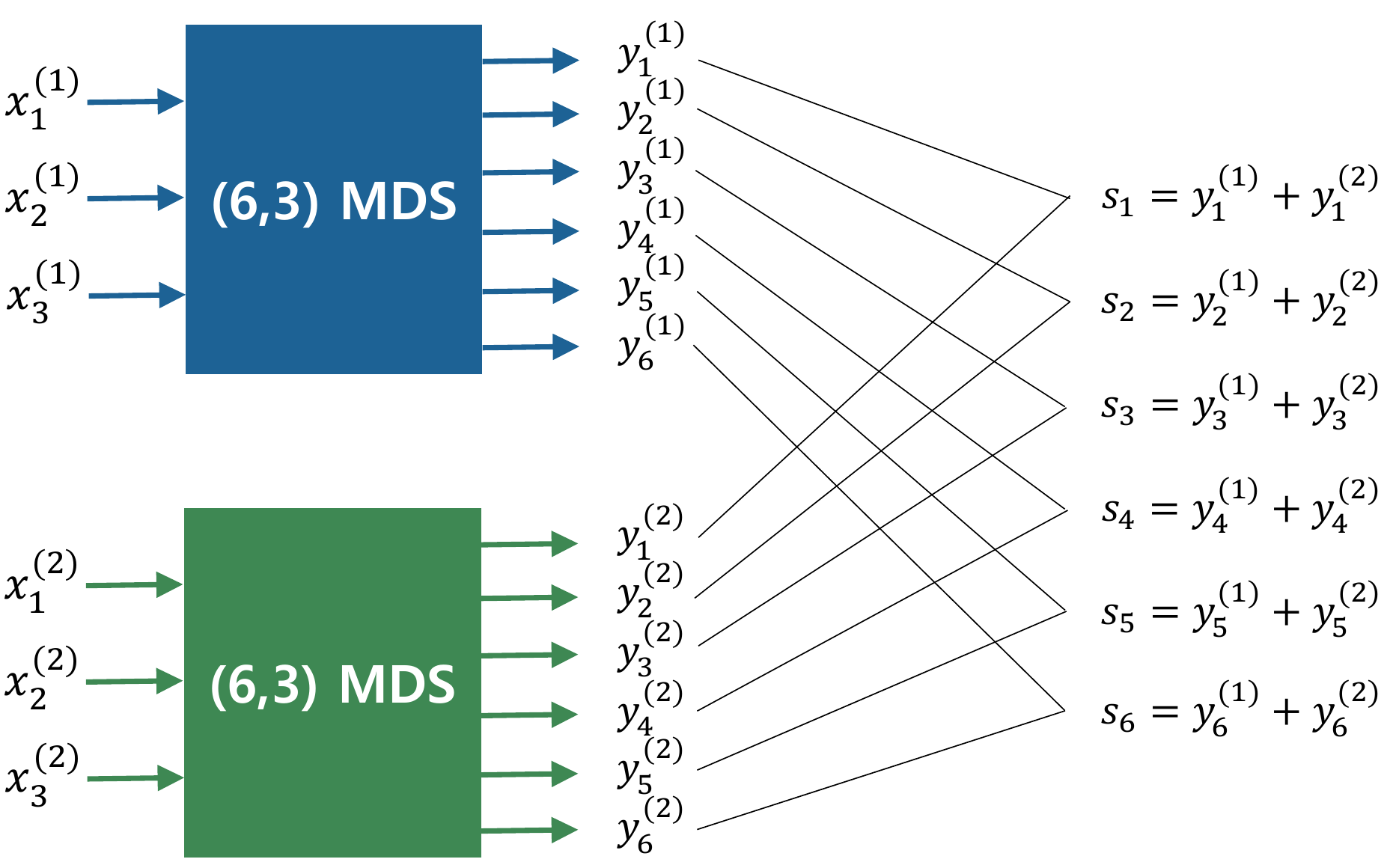}\label{Fig:MDS_Precoding}}
	\quad \quad
	\subfloat[][Allocation of coded symbols into $n$ nodes]{\includegraphics[width=90mm ]{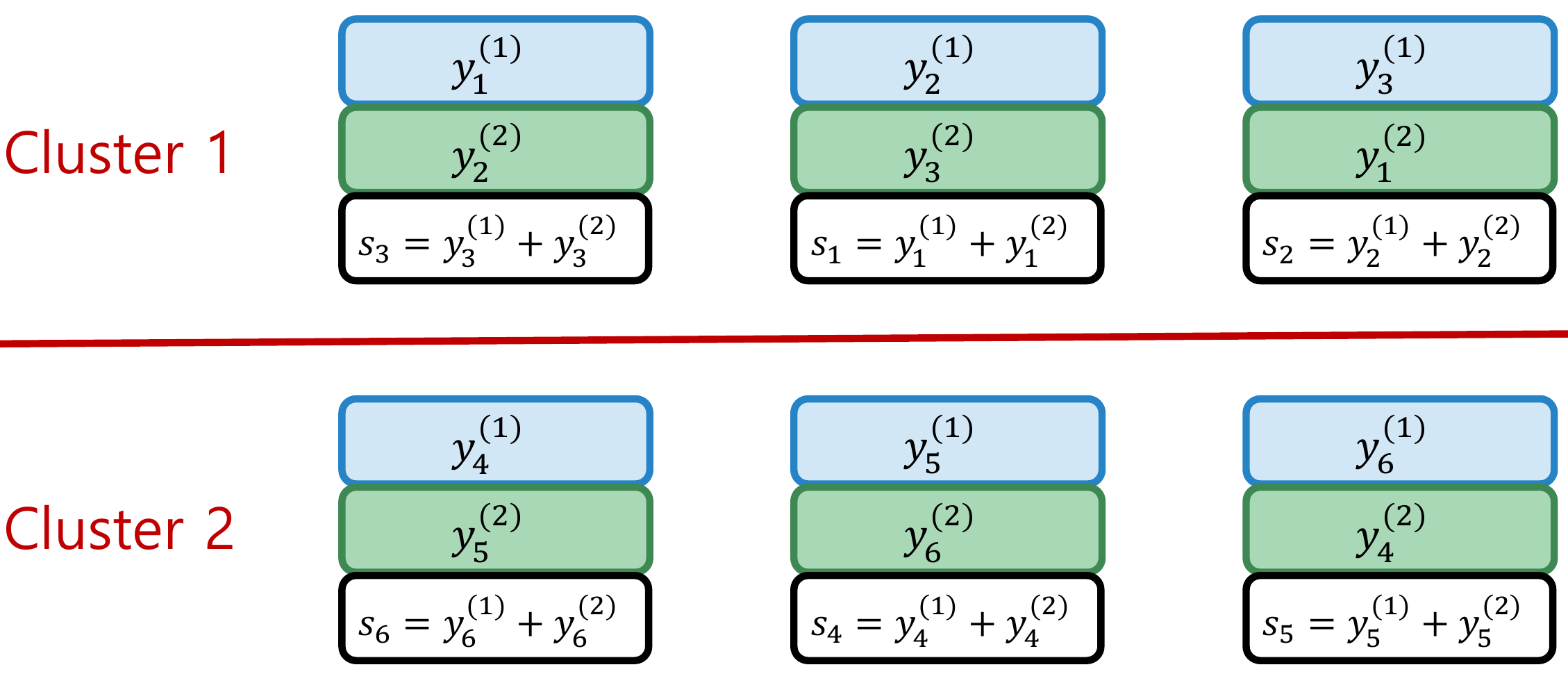}\label{Fig:Alloc_nodes}}
	\caption{MSR code for $\epsilon=0$ with $n_I \divides k$ for $n=6, k=3, L=2$. The construction rule follows the instruction in \cite{papailiopoulos2014locally}, with the concept of the \textit{repair group} in \cite{papailiopoulos2014locally} interpreted as the \textit{cluster} of the present paper.} 
	\label{Fig:MSR_for_epsilon0_divisible}
\end{figure}

\subsection{Code Construction for $n_I \notdivides k$}

Here we construct an MSR code when the given system parameters satisfy $n_I \notdivides k$. The theorem below shows that the optimal $(n,k-q,n_I-1)-$LRC designed in \cite{tamo2016optimal} is a valid MSR code when $n_I \notdivides k$ holds.

\begin{theorem}[MSR Code Construction for $\epsilon=0, n_I \notdivides k$]\label{Thm:LRC2_achieves_MSR}
	Let $\mathds{C}$ be the $(n_0,k_0,r_0)-$LRC constructed in \cite{tamo2016optimal} for $n_0=n, k_0=k-q$ and $r_0=n_I-1$.
	Consider allocating the coded symbols of $\mathds{C}$ in a $[n,k,L]-$clustered DSS, where $r+1=n_I$ nodes within the same repair group of $ \mathds{C}$ are located in the same cluster.
	Then, $\mathds{C}$ is an MSR code for the $[n,k,L]-$clustered DSS under the conditions of $\epsilon=0$ and $n_I \notdivides k$.
\end{theorem}
\begin{proof}
	See Appendix \ref{Section:proof_of_LRC2_achieves_MSR}.
\end{proof}

\begin{figure}
	\centering
	\subfloat[][Encoding structure]{\includegraphics[width=40mm ]{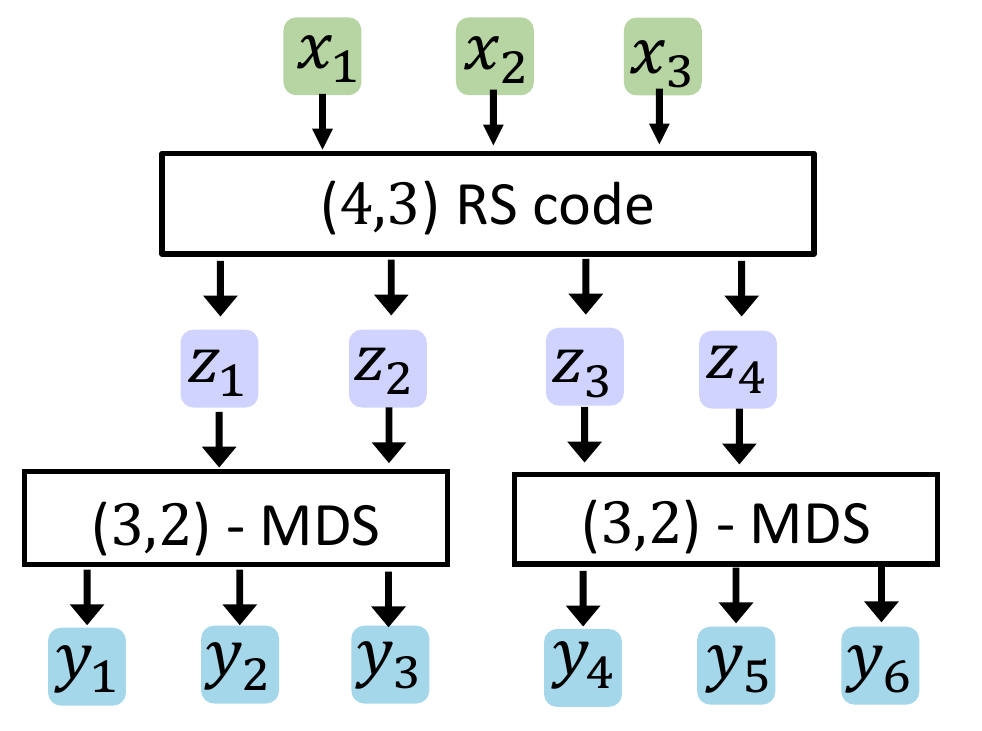}\label{Fig:Epsilon0_nondivisible_ENC}}
	\quad \quad
	\subfloat[][Allocation of coded symbols into $n$ nodes]{\includegraphics[width=40mm ]{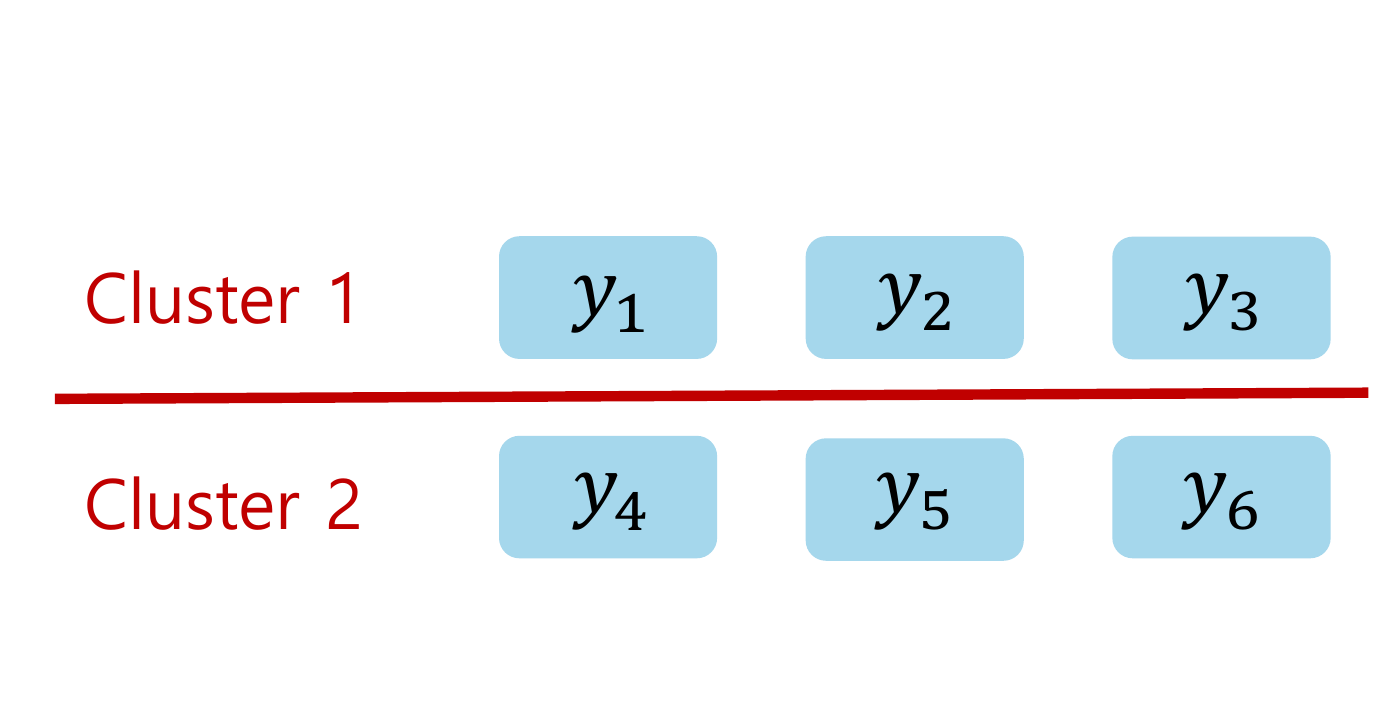}\label{Fig:Epsilon0_nondivisible_alloc}}
	\caption{MSR code for $\epsilon=0$ with $n_I \notdivides k$ case ($n=6, k=4, L=2$). The encoding structure follows from the instruction in \cite{tamo2016optimal}, which constructed $[n_0,k_0,r_0]-LRC$. This paper utilizes $[n,k-q,n_I-1]-LRC$ to construct MSR code for $[n,k,L]$ clustered DSS, in the case of $\epsilon=0$ with $n_I \notdivides k$.}
	\label{Fig:Epsilon0_nondivisible_construct}
\end{figure}

Fig. \ref{Fig:Epsilon0_nondivisible_construct} illustrates an example code construction for the $n_I \notdivides k$ case. Without losing generality, we consider the $\alpha=1$ case; applying this code for $\alpha$ times in a parallel manner achieves the MSR point for the general case with arbitrary positive integer $\alpha$.
In the $[n=6,k=4,L=2]$ clustered DSS with $\epsilon=0$, system parameters are
\begin{align*}
[n_0,k_0,r_0]&=[n,k-q,n_I-1]=[6,3,2],\\
\alpha&=1, \\
\mathcal{M}&=(k-q)\alpha=(k-\floor{k/n_I})=3,
\end{align*}
from Proposition \ref{Prop:parameter_for_small_epsilon}.
The code in Fig. \ref{Fig:Epsilon0_nondivisible_construct} satisfies 1) \textit{exact regeneration} and 2) \textit{data reconstruction} properties, i.e., 
\begin{enumerate}
	\item Any failed node can be exactly regenerated by contacting $n_I-1 = 2$ nodes in the same cluster,
	\item Contacting any $k=4$ nodes can recover the original file $\{x_i: i \in [3]\}$ of size $\mathcal{M}=3$,
\end{enumerate} 
which can be shown as below.
Note that $\{y_i\}_{i=1}^3$ in Fig. \ref{Fig:Epsilon0_nondivisible_construct} is a set of coded symbols generated by a $(3,2)-$MDS code, and this statement also holds for $\{y_i\}_{i=4}^6$. This proves the first property. The second property is directly from the result of \cite{tamo2016optimal}, which states that the minimum distance of the $[n_0,k_0,r_0]-LRC$ is
\begin{align}
d&=n_0-k_0-\left\lceil\dfrac{k_0}{r_0}\right\rceil+2 =6-3-\ceil{3/2}+2=3. 
\end{align}
Note that the $[n_0,k_0,r_0]-LRC$ is already suggested by the authors of \cite{tamo2016optimal}, while the present paper proves that applying this code with $n_0=n, k_0=k-q, r_0=n_I-1$ achieves the MSR point of the $[n,k,L]-$clustered DSS, in the case of $\epsilon=0$ and $n_I \notdivides k$.

\section{MSR Code Design for $ \frac{1}{n-k} \leq \epsilon \leq 1 $}\label{Section:MSR_epsilon_positive}

In the proposition below, we first provide an explicit form of the system parameters when $\frac{1}{n-k} \leq \epsilon \leq 1$. 
Without losing generality, we set the cross-cluster repair bandwidth as $\beta_c=1$. 
In general cases for arbitrary positive integers $\beta_c > 1$, we can apply the code for $\beta_c = 1$ in a parallel manner.
Moreover, we assume that $\beta_I = 1/\epsilon$ is a positive integer.

\begin{prop}\label{Prop:parameter_for_large_epsilon}
	The MSR point $(\alpha, \gamma) = (\alpha_{\text{msr}}^{(\epsilon)}, \gamma_{\text{msr}}^{(\epsilon)})$ for $\frac{1}{n-k}\leq \epsilon \leq 1$ is 
	\begin{equation}\label{Eqn:parameters_for_large_epsilon}
	(\alpha_\text{msr}^{(\epsilon)}, \gamma_\text{msr}^{(\epsilon)}) = \left(\frac{\mathcal{M}}{k}, \frac{\mathcal{M}}{k} \cdot \frac{n-n_I+(n_I-1)/\epsilon}{n-k}\right).
	\end{equation}
	This point satisfies
	$\alpha=n-k$, $\beta_I = 1/\epsilon$ and $\mathcal{M} = k(n-k)$. 
\end{prop}
\begin{proof}
	See Appendix \ref{Section:proof_of_prop_param_large_epsilon}.
\end{proof}

Under this setting, we design MSR codes as below.

\begin{theorem}[MSR Code Construction for $\frac{1}{n-k} \leq \epsilon \leq 1$]\label{Thm:MSR_large_epsilon}
	Let $\mathds{C}$ be an existing MSR code for $[n,k,d=n-1]$ non-clustered DSS\footnote{Here, $[n,k,d]$ non-clustered DSS represents a conventional DSS \cite{dimakis2010network} with $\beta_I=\beta_c$ (i.e. $\epsilon=1$) satisfying the following: 
	contacting any $k$ out of $n$ nodes suffices to recover the source symbols, and any failed node can be regenerated by contacting arbitrary $d$ helper nodes.}.
	Consider a coding scheme $\mathds{C}'$ (a modified version of $\mathds{C}$) defined as follows, which can be applied to an $[n,k,L]-$clustered DSS: 
	\begin{itemize}
		\item First, apply the encoding rule of $\mathds{C}$ to the given source symbols. 
		\item Second, follow the rule for allocating coded symbols into $n$ nodes as specified in $\mathds{C}$. Here, we do not care about which cluster each node resides in.
		\item Consider regenerating a failed node in the $l^{th}$ cluster. Regarding the helper nodes in other clusters, follow the repair rule of $\mathds{C}$.
		As for the helper node in the $l^{th}$ cluster, each node sends the symbol specified in $\mathds{C}$, repeatedly  for $\beta_I=1/\epsilon $ times.
	\end{itemize}
	
Then, applying code $\mathds{C}'$ 
to $[n,k,L]-$clustered DSS achieves the MSR point for $\frac{1}{n-k} \leq \epsilon \leq 1$.
\end{theorem}
\begin{proof}
	See Appendix \ref{Section:proof_of_MSR_large_epsilon}.
\end{proof}

\begin{remark}\label{Rmk:NC_MSR_ex}
Note that in the code construction introduced in Theorem \ref{Thm:MSR_large_epsilon}, we can use any existing MSR code $\mathds{C}$ for $[n,k,d=n-1]$ non-clustered DSS. For example, the MSR codes suggested in \cite{rashmi2011optimal, suh2011exact} can be utilized in the construction.
\end{remark}
	

An intuitive explanation on the result of Theorem \ref{Thm:MSR_large_epsilon} is as follows. Note that the maximum reliably storable file size $\mathcal{M}=k(n-k)$, node storage capacity $\alpha = n-k$ and cross-cluster repair bandwidth $\beta_c = 1$ are invariant to $\epsilon$ with $\frac{1}{n-k} \leq \epsilon \leq 1$,
and only $\beta_I = 1/\epsilon$ varies as $\epsilon$ changes. 
Thus, an existing MSR code for non-clustered DSS with $\epsilon=1$ can be used in the construction of an MSR code for clustered DSS with 
$\frac{1}{n-k} \leq \epsilon \leq 1$; the only modification is needed in $\beta_I$, increasing $\beta_I = 1$ to $\beta_I = 1/\epsilon$. This modification can be done by sending redundant information in the intra-cluster communication link with redundancy $1/\epsilon$.  


%
%
%
%
%

\subsection{Code Construction for $\epsilon = \frac{1}{n-k}, n=kL$}

Here we provide another MSR code construction in Algorithm 
\ref{Algo:MSR_nkL}, which requires a smaller\footnote{A detailed comparison on the required field size is 
	given in Remark \ref{Rmk:MSR_FieldSizeCompare}. }
 field size compared to the construction in Theorem \ref{Thm:MSR_large_epsilon}. 
Note that the code suggested in Algorithm 
\ref{Algo:MSR_nkL} is applicable when $\epsilon = \frac{1}{n-k}$ and $n=kL$ hold.
Here, the system parameters are set to
\begin{align*}
\mathcal{M} &= k(n-k), \quad \quad \quad \quad
\alpha = n-k, \\
\beta_I &= 1/\epsilon = n-k, 
\quad \quad \ \beta_c = 1,
\end{align*}
according to Proposition \ref{Prop:parameter_for_large_epsilon}.
Moreover, the $j^{th}$ node in the $l^{th}$ cluster is denoted as $N_{(l-1)n_I + j}$ in this algorithm, i.e., we have $N_{(l-1)n_I + j} = N(l,j)$ for $l \in [L], j \in [n_I]$.



\begin{algorithm}
	\caption{MSR code construction for $\epsilon = \frac{1}{n-k}, n=kL$}
	\label{Algo:MSR_nkL}
	\begin{algorithmic}
		\REQUIRE System parameters $n, k, L$ 
		\\ \hspace{6.5mm} Source symbol vector $\mathbf{s} = [s_1, \cdots, s_{k(n-k)}]$ 
		\ENSURE Symbols stored on nodes $N_1, \cdots, N_n$
		\STATE \textbf{Step 1.} Generate encoded symbols $c_1, \cdots, c_{n(n-k)}$:
		\FOR{ $i=1,\cdots,n-k $}
		\STATE \hspace{2mm} Apply an $(n,k)-$MDS code (denoted as $\mathds{C}_i$) to $k$ source symbols $s_{k(i-1)+1}, \cdots, s_{ki}$, which generates $n$ coded symbols $c_{n(i-1)+1}, \cdots, c_{ni}$
		\ENDFOR
		
		\STATE \textbf{Step 2.} Distribute coded symbols $\{c_u\}_{u=1}^{n(n-k)}$ 
		to $n$ nodes:
		\STATE \hspace{2mm} Coded symbol $c_u$ is stored in node $N_{mod(u-1,n)+1}$.
		
	\end{algorithmic}
\end{algorithm}

The code suggested in Algorithm \ref{Algo:MSR_nkL} is illustrated in Fig. \ref{Fig:codekL}. 
This code has the following property.

\begin{figure}[!t]
	\centering
	\includegraphics[height=40mm]{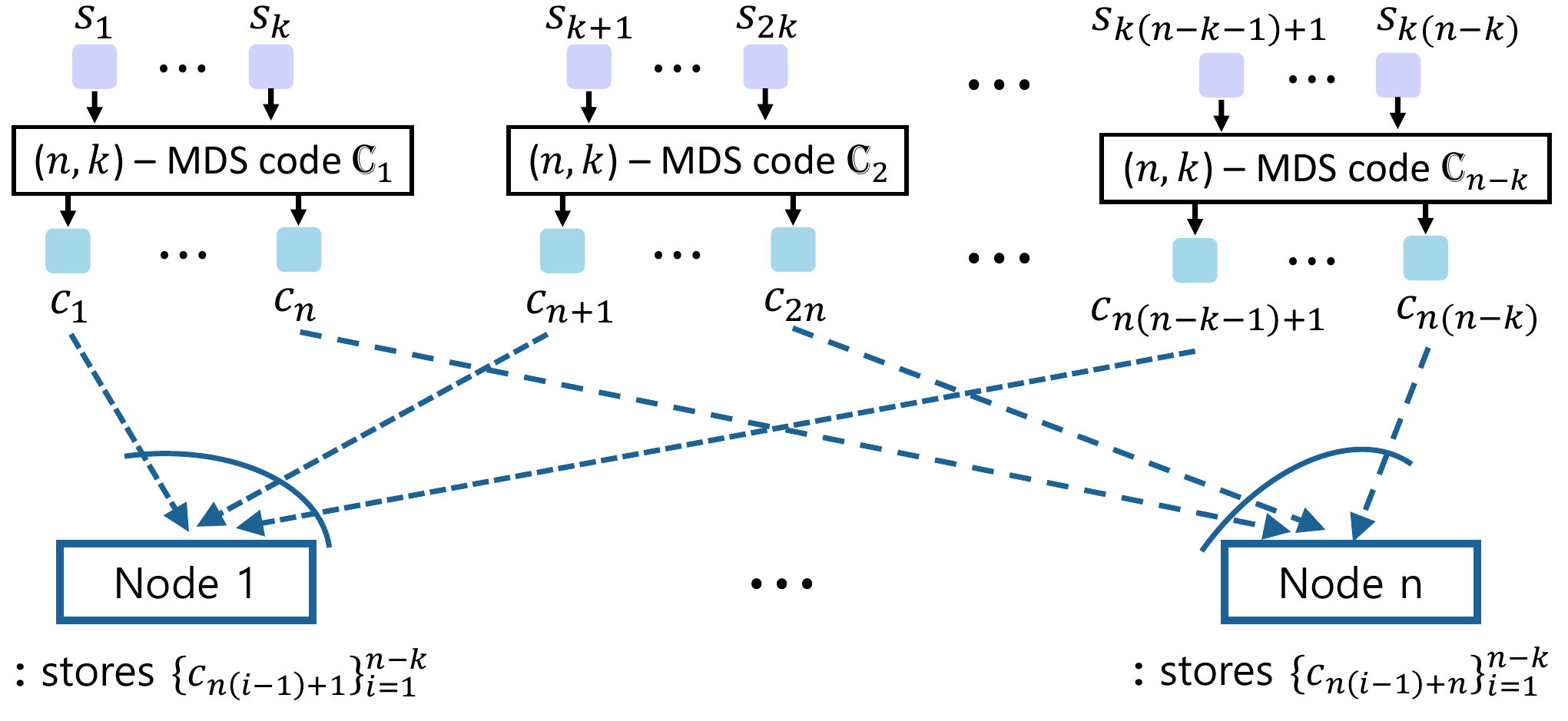}
	\caption{MSR code construction suggested in Algorithm \ref{Algo:MSR_nkL}. }
	\label{Fig:codekL}
\end{figure}

\begin{remark}\label{Rmk:MSRkL}
The code suggested in Algorithm \ref{Algo:MSR_nkL} satisfies the following, which can be confirmed from Fig. \ref{Fig:codekL}:
\begin{itemize}
	\item Each node contains $\alpha=n-k$ coded symbols, which are generated from distinct $(n,k)-$MDS codes. To be specific, node $N_t$ stores $(n-k)$ coded symbols $\{c_{n(i-1)+t}\}_{i=1}^{n-k}$, where $c_{n(i-1)+t}$ is generated from $(n,k)-$MDS code $\mathds{C}_i$.
\end{itemize}
\end{remark}

Moreover, the code suggested in Algorithm \ref{Algo:MSR_nkL} is a valid MSR code when $\epsilon=\frac{1}{n-k}$ and $n=kL$, as stated below.

\begin{theorem}
	\label{Thm:MSR_nkL}
	The code in Algorithm \ref{Algo:MSR_nkL} is an MSR code for $[n,k,L]-$clustered DSS with $\epsilon=\frac{1}{n-k}$ and $n=kL$. In other words, it satisfies all requirements stated in Condition \ref{Condition:code_general}:
	\begin{itemize}
		\item 
		Each node contains $\alpha=n-k$ coded symbols.
		\item \textbf{(Exact regeneration)} When a node fails, it can be exactly regenerated by using the intra-cluster repair bandwidth of $\beta_I=n-k$ and the cross-cluster repair bandwidth of $\beta_c=1$. 
		\item \textbf{(Data reconstruction)} Contacting any $k$ out of $n$ nodes can retrieve vector $\mathbf{s}$ with $\mathcal{M}=k(n-k)$ source symbols.
	\end{itemize}
%
\end{theorem}
\begin{proof}
The proof is in Appendix \ref{Section:proof_of_MSR_nkL}.
\end{proof}

Fig. \ref{codekL} illustrates an example of the code suggested in Algorithm \ref{Algo:MSR_nkL}, when $n=6$, $k=2$, $L=3$, and $\epsilon=\frac{1}{n-k}=\frac{1}{4}$.
In order to store $\mathcal{M}=k(n-k)=8$ source symbols $\{u_i\}_{i=1}^8$, four $(6,2)$ MDS codes are used to generate $n(n-k)=24$ encoded symbols $\{c_i\}_{i=1}^{24}$. Afterwards, the coded symbols are allocated as described in Algorithm \ref{Algo:MSR_nkL}.
Now, we show that this code satisfies the properties in Theorem \ref{Thm:MSR_nkL}. First, each node contains $\alpha = n-k = 4$ coded symbols. Secondly, suppose a node fails; as an example illustrated in Fig. \ref{codekL2}, let $N(1,1)$ fail. Then, $n_I - 1 = 1$ node in the same cluster (with the failed node) transmits $\beta_I=4$ coded symbols, $c_2, c_8, c_{14}, c_{20}$. Moreover, $n-n_I = 4$ nodes in other clusters transmit $\beta_c =1$ symbol each, corresponding to $c_3, c_{10}, c_{17}, c_{24}$, respectively.
Using the received $\gamma = 8$ coded symbols, we can exactly regenerate the failed node as follows. Note that the source symbols $u_1, u_2$ can be recovered from received coded symbols $c_2, c_3$, by decoding a $(6,2)-$MDS code. In a similar way, we can recover all source symbols $\{u_i\}_{i=1}^8$. Thus, we can exactly regenerate $c_1, c_7, c_{13}, c_{19}$ contained in the failed node using the source symbols.
Finally, we confirm the data reconstruction property. By contacting arbitrary $k=2$ nodes, we obtain two red coded symbols which are generated from the first $(6,2)-$MDS code. Similarly, we obtain two yellow/green/blue coded symbols. Thus, we can successfully decode all four MDS codes, and obtain all source symbols $\{u_i\}_{i=1}^8$. 

\begin{figure}[!h]
	\centering
	\includegraphics[height=38mm]{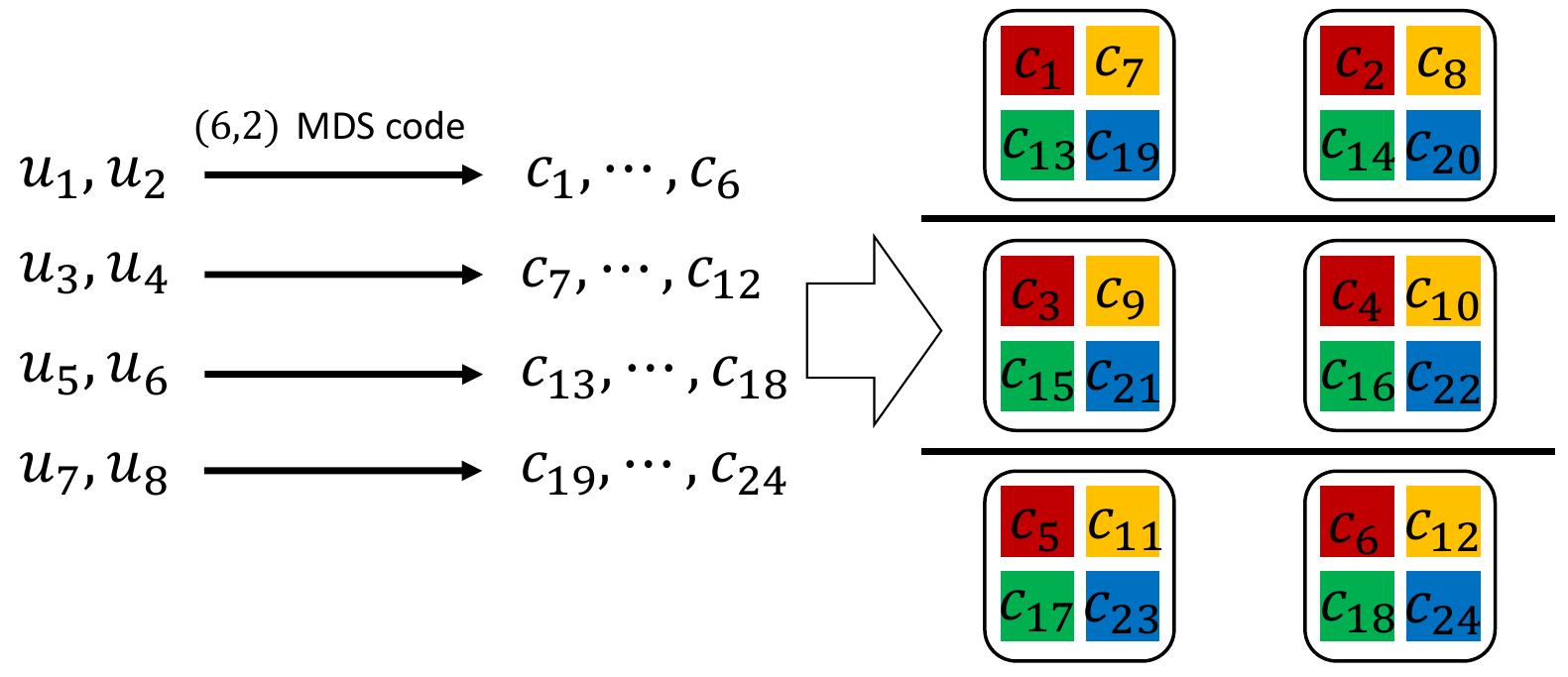}
	\caption{MSR code construction from Algorithm \ref{Algo:MSR_nkL}, when $n=6, k=2, L=3$ and $\epsilon=1/4$.}
	\label{codekL}
\end{figure}

\begin{figure}[!h]
	\centering
	\includegraphics[height=38mm]{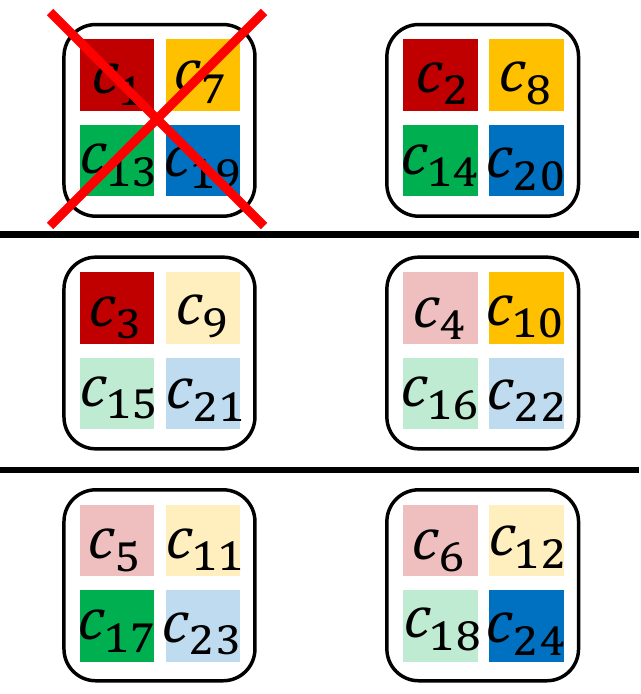}
	\caption{The node repair process of the code in Fig. \ref{codekL}.  When node $N(1,1)$ fails, the node in the same cluster transmits $\beta_I=4$ coded symbols, while nodes in other clusters send $\beta_c = 1$ coded symbol. }
	\label{codekL2}
\end{figure}

\begin{remark}\label{Rmk:MSR_FieldSizeCompare}
Here we compare two MSR codes when $\epsilon=1/(n-k)$ and $n=kL$ hold: one is a modified version of an existing code\footnote{According to Theorem \ref{Thm:MSR_large_epsilon} and Remark \ref{Rmk:NC_MSR_ex}, the MSR code in \cite{suh2011exact} can be applied to $[n,k,L]-$clustered DSS with $\frac{1}{n-k} \leq \epsilon \leq 1$, with a slight modification specified in Theorem \ref{Thm:MSR_large_epsilon}.} in \cite{suh2011exact} and the other is the code suggested in Algorithm \ref{Algo:MSR_nkL}. The required field size of these codes are compared as follows.
Recall that the code in \cite{suh2011exact} can be applied under the setting of $n \geq 2k$. Since we consider the case of $n=kL$, the constraint above is equivalent to $L\geq 2, n=kL$. 
Note that the code in \cite{suh2011exact} requires a field size of at most $2(n-k)$, while Algorithm \ref{Algo:MSR_nkL} needs a field size of at most $n$ since it is based on $(n,k)-$MDS codes. Thus, the code in Algorithm \ref{Algo:MSR_nkL} has a smaller field size when $L\geq3$, compared to the code in \cite{suh2011exact}.
\end{remark}


%

\section{Conclusion}
Focusing on clustered distributed storage systems which reflect the structures of real data centers and wireless storage networks, MBR and MSR coding schemes which achieve capacity have been presented.
The proposed coding schemes satisfy two key requirements: 1) exact regeneration of an arbitrary failed node using minimum system resources,
and 2) data reconstruction by contacting arbitrary $k$ nodes. 

In the first part of this paper, MBR codes for arbitrary parameter values of $n,k,L,\beta_I,\beta_c$ are suggested. The construction of the proposed MBR code is based on the repair-by-transfer scheme suggested in \cite{rashmi2009explicit}, while here we modified the scheme to reflect the clustered nature of storage nodes with limited available cross-cluster bandwidths. 
The proposed construction rule depends on the $\epsilon=\beta_c/\beta_I$ regime. When $\epsilon=0$, the coded symbols are shared within each cluster, so that any failed node can be exactly regenerated without using any cross-cluster repair bandwidths.
The MBR code proposed for the $\epsilon = 0$ case requires a much smaller field size, compared to existing local MBR codes in \cite{kamath2014TIT, kamath2013ISIT}. When the cross-cluster repair bandwidths are allowed, i.e., $0 < \epsilon \leq 1$, the coded symbols are partitioned into two subsets called \textit{local symbols} and \textit{global symbols}: the local symbols are shared among nodes in the same cluster, while the global symbols are shared either among the nodes in the same cluster or among the nodes in other clusters. 

In the second part of this paper, MSR codes for clustered distributed storage are suggested for two important cases: 1) the $\epsilon=0$ case which only uses intra-cluster communication for repairing a failed node, and 2) the scenario of $\frac{1}{n-k} \leq \epsilon \leq 1$ which allows achieving the minimum node storage overhead of $\alpha= \mathcal{M}/k $. When $\epsilon=0$, two existing LRCs \cite{papailiopoulos2014locally, tamo2016optimal} are shown to be the MSR codes for clustered distributed storage, under the settings of $n_I \divides k$ and $n_I \notdivides k$, respectively.
When $\frac{1}{n-k} \leq \epsilon \leq 1$, it is shown that utilizing existing MSR codes for non-clustered distributed storage
(with a slight modification in the repair process) achieves the MSR point of clustered distributed storage.
Finally, for the scenarios satisfying $\epsilon=\frac{1}{n-k}$ and $n=kL$, we propose a simple MSR code which requires a small field size than the existing code for $L\geq3$.






\appendices
\numberwithin{equation}{section}

\section{Proof of Theorem \ref{Thm:CodeZero}}\label{Proof:Thm:CodeZero}

The first statement is directly proved by using Lemma \ref{Prop:MBR for zero gammac}\textendash\ref{Prop:MBR for zero gammac_fourth}. Now we move on to the second statement for the exact regeneration of a failed node. Suppose that $N(l_0,j_0)$, the $j_0^{th}$ storage node in the $l_0^{th}$ cluster, is out of order. From Lemma \ref{Prop:MBR for zero gammac}\textendash\ref{Prop:MBR for zero gammac_first} and Lemma \ref{Prop:MBR for zero gammac}\textendash\ref{Prop:MBR for zero gammac_second}, any coded symbol in $N(l_0,j_0)$ is also stored in another node in the $l_0^{th}$ cluster. In other words, the set $\cup_{j\neq j_0} \{N(l_0,j)\}$ contains all symbols stored in $N(l_0,j_0)$. Thus, $N(l_0,j_0)$ can be exactly regenerated by contacting other nodes in the $l_0^{th}$ cluster only.
From Lemma \ref{Prop:MBR for zero gammac}\textendash\ref{Prop:MBR for zero gammac_third}, we can confirm that the node $N(l_0,j_0)$ can be exactly regenerated by downloading $\beta_I = 1$ symbol from each of the $n_I-1$ nodes in $\cup_{j\neq j_0} \{N(l_0,j)\}$. Moreover, the regeneration process does not incur cross-cluster repair traffic, i.e., $\beta_c = 0$.

We now prove the third statement on the data reconstruction property.
Consider a data collector (DC) which connects to arbitrary $k$ out of $n$ nodes to recover the original source symbol vector $\mathbf{s} = [s_1, s_2, \cdots, s_\mathcal{M}]^T$. Here, we prove that contacting arbitrary $k$ nodes can obtain $\mathcal{M}$ coded symbols $\{c_i\}$, which are sufficient for recovering the vector $\mathbf{s} $ of $\mathcal{M}$ source symbols according to the MDS property of the $[\theta, \mathcal{M}]$ code used in Algorithm \ref{Algo:MBR_code_zero}.
For given $k$ nodes contacted by DC, define the corresponding \textit{contact vector} $\boldsymbol{\omega} = [\omega_1, \cdots, \omega_L]$ where $\omega_l$ represents the number of contacted nodes in the $l^{th}$ cluster. Then, the set of possible contact vectors is expressed as 
\begin{equation}
\Omega= \left\{\boldsymbol{\omega} = [\omega_1, \cdots, \omega_L]: \sum_{l=1}^L \omega_l = k, \omega_l \in \{0, 1, \dots, n_I\} \right\}.
\end{equation}
Let $n(\boldsymbol{\omega})$ be the number of distinct coded symbols obtained by contacting $k$ nodes with the corresponding contact vector being $\boldsymbol{\omega}$.
Then, we establish the following bound on $n(\boldsymbol{\omega})$.
\begin{lemma}\label{Prop:Lower_bound on the number of retrievable coded symbols for zero gammac}	
	Consider an $[n,k,L]-$clustered DSS with the code in Algorithim \ref{Algo:MBR_code_zero} applied. 
	Let a DC contacts arbitrary $k$ nodes with the corresponding contact vector of $\boldsymbol{\omega}$. Then, the number of distinct coded symbols $\{c_i\}$ retrieved by DC is lower bounded by $\mathcal{M}$ in (\ref{Eqn:Capacity for gamma_c = 0}). In other words, 
	\begin{equation}\label{Eqn:Lower_bound on the number of retrievable coded symbols_gammac0 case}
	n(\boldsymbol{\omega}) \geq \mathcal{M}  \ \ \   \forall \boldsymbol{\omega}\in \Omega.
	\end{equation}
\end{lemma}
\begin{proof}
	See Appendix \ref{Proof:Lower bound on the number of retrievable coded symbols_gammac0 case}.
\end{proof}
Therefore, for an arbitrary contact of $k$ nodes, the suggested coding scheme guarantees at least $\mathcal{M}$ distinct coded symbols to be retrieved. Using the MDS property of the $[\theta, \mathcal{M}]$ code in Algorithm \ref{Algo:MBR_code_zero}, we can confirm that the original source symbol $\mathbf{f}$ can be obtained from the retrieved $\mathcal{M}$ distinct coded symbols. This completes the proof for the data reconstruction property.

\section{Proof of Theorem \ref{Thm:CodeNonzero}}\label{Proof:Thm:CodeNonzero}

The first statement is obtained directly from Lemma \ref{Prop:MBR}\textendash\ref{Prop:MBR_fourth} and the definition of $\chi= 1/\epsilon$. The second statement which claims the exact regeneration property is proved as follows.
Consider $N(l_0,j_0)$, the $j_0^{th}$ storage node in the $l_0^{th}$ cluster, is broken. 
From Lemma \ref{Prop:MBR}\textendash\ref{Prop:MBR_first} and Lemma \ref{Prop:MBR}\textendash\ref{Prop:MBR_third}, each survived node in the $l_0^{th}$ cluster contains $\chi$ distinct coded symbols which are stored in  $N(l_0,j_0)$. Therefore, the set $\cup_{j\neq j_0} \{N(l_0,j)\}$ contains $(n_I-1)\chi$ symbols stored in $N(l_0,j_0)$.
Similarly, from Lemma \ref{Prop:MBR}\textendash\ref{Prop:MBR_first} and Lemma \ref{Prop:MBR}\textendash\ref{Prop:MBR_second}, each survived node \textit{not} in the $l_0^{th}$ cluster contains one distinct coded symbol which is stored in $N(l_0,j_0)$. Therefore, the set $\cup_{l\neq l_0} \{N(l,j)\}$ contains $(n-n_I)$ symbols stored in $N(l_0,j_0)$.
In summary, $\alpha = (n_I-1)\chi + (n-n_I)$ coded symbols stored in the failed node $N(l_0,j_0)$ can be recovered by contacting $n_I-1$ nodes in the $l_0^{th}$ cluster and $n-n_I$ nodes in other ($l \neq l_0$) clusters; $n_I-1$ nodes within the same cluster transmit $\chi$ coded symbols each, while $n-n_I$ nodes in other clusters contribute one coded symbol each. Therefore, this process satisfies $\beta_I = \chi, \beta_c = 1$, as described in Section \ref{Section:parameter for MBR}. Moreover, the total repair bandwidth is expressed as $\gamma = (n_I-1)\chi+(n-n_I) = (n_I-1)/\epsilon + (n-n_I) = \gamma_{\text{mbr}}^{(\epsilon)}$.

Finally, the third statement for data reconstruction is proved as follows. We here use the notations $n(\boldsymbol{\omega})$ and $\Omega$, which are defined in Appendix \ref{Proof:Thm:CodeZero}.
Similar to Lemma \ref{Prop:Lower_bound on the number of retrievable coded symbols for zero gammac}, we have the following Lemma for the case of $0 < \epsilon \leq 1$. This Lemma below completes the proof for the data reconstruction property, in a similar way that Lemma \ref{Prop:Lower_bound on the number of retrievable coded symbols for zero gammac} completes the proof of Theorem \ref{Thm:CodeZero} for $\epsilon=0$.

\begin{lemma}\label{Prop:Lower_bound on the number of retrievable coded symbols}	
	Consider an $[n,k,L]-$clustered DSS with the code in Algorithm \ref{Algo:CodeNonzero} applied. Let a DC contacts arbitrary $k$ nodes with $\boldsymbol{\omega}$ being the corresponding contact vector. Then, the number of distinct coded symbols $\{c_i\}$ retrieved by DC is lower bounded by $\mathcal{M}$ in (\ref{Eqn:Capacity for MBR}). In other words, 
	\begin{equation}\label{Eqn:Lower_bound on the number of retrievable coded symbols}
	n(\boldsymbol{\omega}) \geq \mathcal{M}  \ \ \   \forall \boldsymbol{\omega}\in \Omega.
	\end{equation}
\end{lemma}
\begin{proof}
	See Appendix \ref{Proof:Lower bound on the number of retrievable coded symbols}.
\end{proof}

\section{Proof of Theorem \ref{Thm:LRC1_achieves_MSR}}\label{Section:proof_of_LRC1_achieves_MSR}

We focus on code $\mathds{C}$, the explicit ($n,r,d,\mathcal{M},\alpha$)-LRC constructed in Section V of \cite{papailiopoulos2014locally}. This code has the parameters
\begin{equation}\label{Eqn:LRC1_param}
(n,r,d=n-k+1,\mathcal{M},\alpha=\frac{r+1}{r}\frac{\mathcal{M}}{k}),
\end{equation}
where $r$ is the repair locality and $d$ is the minimum distance, and other parameters ($n,\mathcal{M},\alpha$) have physical meanings identical to those in the present paper.
By setting $r=n_I-1$, the code has a node capacity of
\begin{equation}\label{Eqn:alpha_MSR_small_epsilon}
\alpha = \frac{n_I}{n_I-1}\frac{\mathcal{M}}{k} = \frac{\mathcal{M}}{k(1-1/n_I)}=\frac{\mathcal{M}}{k-q}
\end{equation}
where the last equality holds from the $n_I \divides k$ condition and the definition of $q$ in (\ref{Eqn:quotient}).

We first prove that any node failure can be exactly regenerated by using the system parameters in \eqref{Eqn:parameters_for_small_epsilon}.
According to the description in Section V-B of \cite{papailiopoulos2014locally}, \textit{any} node is contained in a unique corresponding repair group of size $r+1=n_I$, so that a failed node can be exactly repaired by contacting $r=n_I-1$ other nodes in the same repair group. This implies that a failed node does not need to contact other repair groups in the exact regeneration process. 
By setting each repair group as a cluster (note that each cluster contains $n_I=n/L$ nodes), we can achieve 
\begin{equation}\label{Eqn:betac_MSR_small_epsilon}
\beta_c = 0.
\end{equation}
Moreover, Section V-B of \cite{papailiopoulos2014locally} illustrates that the exact regeneration of a failed node is possible by contacting the \textit{entire} symbols contained in $r=n_I-1$ nodes in the same repair group, and applying the XOR operation. This implies $\beta_I = \alpha$, which result in
\begin{equation}\label{Eqn:betai_MSR_small_epsilon}
\gamma = (n_I-1)\beta_I = (n_I-1) \frac{\mathcal{M}}{k-q},
\end{equation}
combined with \eqref{Eqn:gamma} and \eqref{Eqn:alpha_MSR_small_epsilon}.
From (\ref{Eqn:alpha_MSR_small_epsilon}) and (\ref{Eqn:betai_MSR_small_epsilon}), we can conclude that code $\mathds{C}$ satisfies the exact regeneration of any failed node using the parameters in \eqref{Eqn:parameters_for_small_epsilon}.

Now we prove that contacting any $k$ nodes suffices to recover original data in the clustered DSS with code $\mathds{C}$ applied. Note that the minimum distance is $d=n-k+1$ from (\ref{Eqn:LRC1_param}). Thus, the information from $k$ nodes suffices to pick the correct codeword. This completes the proof of Theorem \ref{Thm:LRC1_achieves_MSR}.

\section{Proof of Theorem \ref{Thm:LRC2_achieves_MSR}}\label{Section:proof_of_LRC2_achieves_MSR}

We first prove that code $\mathds{C}$ has a minimum distance of $d=n-k+1$, which implies that the original file of size $\mathcal{M}=k-q$ can be recovered by contacting arbitrary $k$ nodes. Second, we prove that any failed node can be exactly regenerated under the setting of (\ref{Eqn:parameters_for_small_epsilon}).
Recall that the $[n_0,k_0,r_0]-$LRC constructed in \cite{tamo2016optimal} has the following property:

\begin{lemma}[Theorem 1 of \cite{tamo2016optimal}]\label{Lemma:Result_of_Tamo}
	The code constructed in \cite{tamo2016optimal} has locality $r_0$ and optimal minimum distance $d=n_0-k_0-\ceil{\frac{k_0}{r}}+2$, when $(r_0+1)\divides n_0$.
\end{lemma}
Note that we consider code $\mathds{C}$ of optimal $[n_0,k_0,r_0]=[n,k-q,n_I-1]-$LRC. Since $r_0+1=n_I$ divides $n_0=n$, Lemma \ref{Lemma:Result_of_Tamo} can be applied. 
The result of Lemma \ref{Lemma:Result_of_Tamo} implies that the minimum distance of $\mathds{C}$ is
\begin{align}\label{Eqn:dmin_tamo}
d &= n-(k-q)-\left\lceil\dfrac{k-q}{n_I-1}\right\rceil+2.
\end{align}
Since we consider the $n_I \notdivides k$ case, we have
\begin{equation}\label{Eqn:k_case2}
k = qn_I + m, \quad \quad (0 < m \leq n_I-1)
\end{equation}
from (\ref{Eqn:remainder}). 
Inserting (\ref{Eqn:k_case2}) into (\ref{Eqn:dmin_tamo}), we have
\begin{align}\label{dmin_case2}
d &= n-(k-q)-\left\lceil\dfrac{(n_I-1)q+m}{n_I-1}\right\rceil+2 \nonumber\\
&= n-(k-q)-(q+1)+2 = n-k+1,
\end{align}
where the second last equality holds since $0 < m \leq n_I-1$ from (\ref{Eqn:k_case2}). Thus, this proves that contacting arbitrary $k$ nodes suffices to recover the original source file.

\begin{figure}[!t]
	\centering
	\includegraphics[width=90mm]{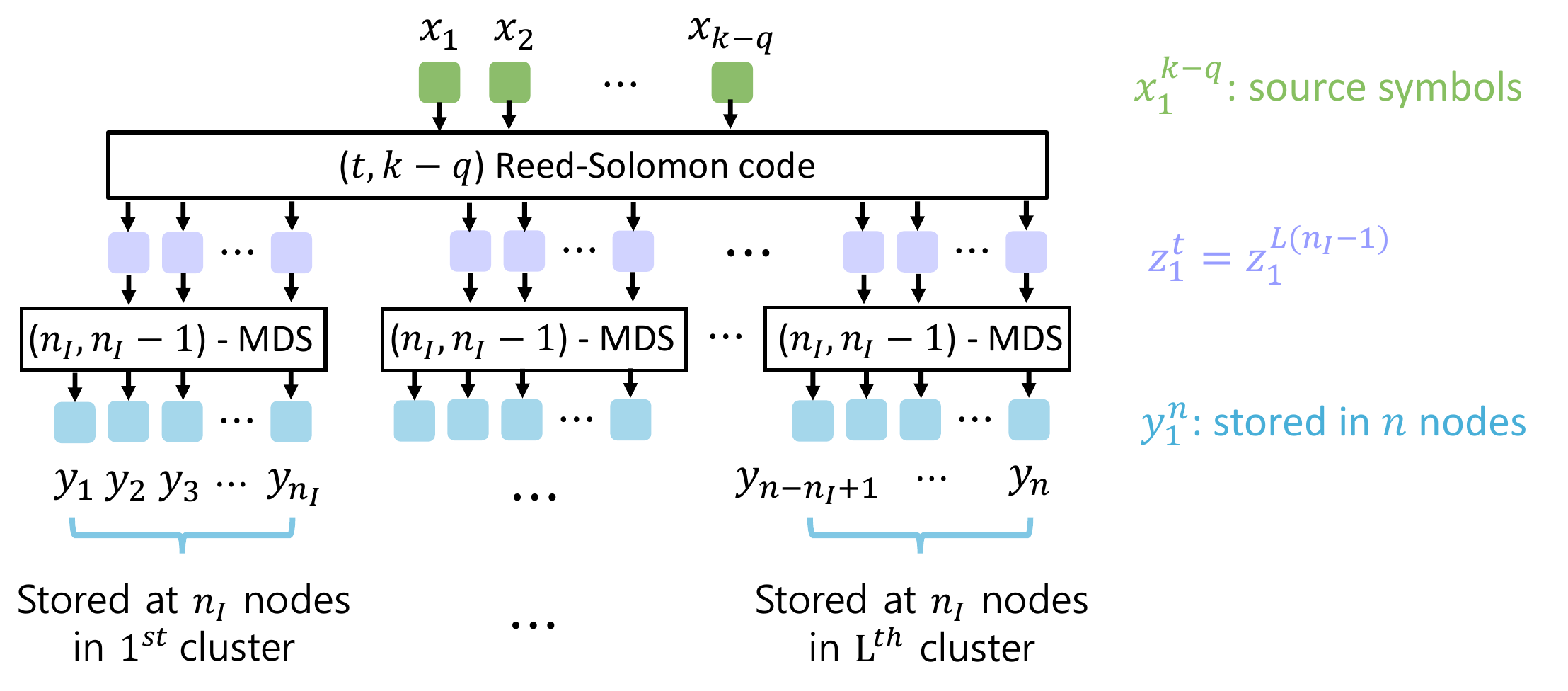}
	\caption{Code construction for $\epsilon=0, n_I \notdivides k$ case}
	\label{Fig:Epsilon_nondivisible}
\end{figure}

Now, all we need to prove is that any failed node can be exactly regenerated under the setting of system parameters specified in Proposition \ref{Prop:parameter_for_small_epsilon}.
According to the rule illustrated in \cite{tamo2016optimal}, the construction of code $\mathds{C}$ can be shown as in Fig. \ref{Fig:Epsilon_nondivisible}.
First, we have $\mathcal{M}=k-q$ source symbols $\{x_i\}_{i=1}^{k-q}$ to store reliably. By applying a $(T,k-q)$ Reed-Solomon code to the source symbols, we obtain $\{z_i\}_{i=1}^{t}$ where $T \coloneqq L(n_I-1)$. Then, we partition $\{z_i\}_{i=1}^T$ symbols into $L$ groups, where each group contains $(n_I-1)$ symbols. Next, each group of $\{z_i\}$ symbols is encoded by an $(n_I, n_I-1)-$MDS code, which result in a group of $n_I$ symbols of $\{y_i\}$. Finally, we store symbol $y_{n_I(l-1)+j}$ in node $N(l,j)$. By this allocation rule, $y_i$ symbols in the same group are located in the same cluster.

Assume that $N(l,j)$, the $j^\text{th}$ node in the $l^\text{th}$ cluster,  containing $y_{n_I(l-1)+j}$ symbol fails
for $l \in [L]$ and $j \in [n_I]$.
From Fig. \ref{Fig:Epsilon_nondivisible}, we know that $(n_I-1)$ symbols of $\{y_{n_I(l-1)+s}\}_{s=1, s\neq j}^{n_I}$ stored in the $l^{th}$ cluster can decode the $(n_I, n_I-1)-$MDS code for group $l$. Thus, the contents of $y_{n_I(l-1)+j}$ can be recovered by retrieving symbols from nodes in the the $l^{th}$ cluster (i.e., the same cluster where the failed node is in). This proves the ability of exactly regenerating an arbitrary failed node. 
The regeneration process satisfies
\begin{equation}\label{Eqn:Epsilon0_Regeneration}
\beta_c = 0, \beta_I = \alpha.
\end{equation}
Moreover, note that the code in Fig. \ref{Fig:Epsilon_nondivisible} has
\begin{equation}\label{Eqn:Epsilon0_capacity}
\mathcal{M} = (k-q)\alpha
\end{equation}
source symbols. 
Since parameters obtained in (\ref{Eqn:Epsilon0_Regeneration}) and (\ref{Eqn:Epsilon0_capacity}) are consistent with Proposition \ref{Prop:parameter_for_small_epsilon}, we can confirm that code $\mathds{C}$ achieves a valid MSR point under the conditions $\epsilon = 0$ and $n_I \notdivides k$.

\section{Proof of Theorem \ref{Thm:MSR_large_epsilon}}\label{Section:proof_of_MSR_large_epsilon}

Here we prove that $\mathds{C}$, an MSR code for $[n,k,d=n-1]$ non-clustered DSS also achieves the MSR point of $[n,k,L]-$clustered DSS with $\epsilon \in [\frac{1}{n-k}, 1]$.
Thus, all we need to check is whether or not code $\mathds{C}$ satisfies Condition \ref{Condition:code_general}.

First, according to Proposition \ref{Prop:parameter_for_large_epsilon}, we have $\alpha_{\text{msr}}^{(\epsilon)} = n-k$ for $\epsilon \in [\frac{1}{n-k}, 1]$. Moreover, code $\mathds{C}$ has $\alpha = \mathcal{M}/k = k(d-k+1)/k = n-k$ as in \cite{suh2011exact}.
Thus, code $\mathds{C}$ satisfies the node storage capacity condition for the MSR point with $\epsilon \in [\frac{1}{n-k}, 1]$.
Second, we check the data reconstruction condition. According to Proposition \ref{Prop:parameter_for_large_epsilon}, we can extract $\mathcal{M}=k(n-k)$ symbols irrespective of $\epsilon \in [\frac{1}{n-k}, 1]$, under the setting of contacting arbitrary $k$ nodes. Thus, code $\mathds{C}$ satisfies the data reconstruction condition. 
Third, we check the exact regeneration condition. Note that when we apply code $\mathds{C}$ to $n$ storage nodes, it is guaranteed that any failed node can be exactly regenerated by contacting $n-1$ helper nodes, while each helper node transmits $\beta = 1$ symbol. Note that the case of $\epsilon \in [\frac{1}{n-k}, 1]$ has a more relaxed exact regeneration condition: the intra-cluster repair bandwidth may increase up to $\beta_I = 1/\epsilon \geq 1$. Thus, sending redundant information in the intra-cluster link (with the redundancy of $1/\epsilon$) is sufficient to achieve the exact regeneration condition.

\section{Proof of Theorem \ref{Thm:MSR_nkL}}\label{Section:proof_of_MSR_nkL}

The first condition on node storage capacity $\alpha=n-k$ is directly confirmed from Remark \ref{Rmk:MSRkL}.
Next, regarding the exact regeneration condition, consider the following repair process when node $N(l,j)$ fails.
\begin{itemize}
	\item \textbf{Intra-cluster transmission:} $(n_I-1)$ survived nodes in the $l^{th}$ cluster (i.e., the cluster which contains the failed node)  transmit $\beta_I = n-k=\alpha$ symbols to the failed node. In other words, they send all the symbols they have to $N(l,j)$.
	\item \textbf{Cross-cluster transmission:} Let $N_1', \cdots N_{n-k}'$ be $n-n_I = n-k$ nodes in other clusters. Recall that for arbitrary $t\in [n-k]$, each node contains a coded symbol generated from $(n,k)-$MDS code $\mathds{C}_t$, according to Remark \ref{Rmk:MSRkL}.
	Set node $N_t'$ to transmit coded symbol generated from code $\mathds{C}_t$ for $t \in [n-k]$. 
\end{itemize}
For each $t\in [n-k]$, the intra-cluster transmission provides 
$n_I - 1 = k -1$ coded symbols generated from $\mathds{C}_t$.
Moreover, the cross-cluster transmission gives 
one coded symbol generated from $\mathds{C}_t$, for each $t \in [n-k]$.
Thus, total $k$ coded symbols (generated from $\mathds{C}_t$) are retrievable by using intra and cross-cluster communications. This can decode $(n-k)$ MDS codes $\{\mathds{C}_t\}_{t=1}^{n-k}$ and exactly regenerate the symbols stored in the failed node $N(l,j)$.


Finally, we show the data reconstruction property. From Remark \ref{Rmk:MSRkL}, we can retrieve 
$k$ coded symbols generated from $\mathds{C}_t$ for each $t \in [n-k]$, by contacting arbitrary $k$ nodes.
Thus, all source symbols $\{s_i\}_{i=1}^{k(n-k)}$ can be reconstructed by decoding $n-k$ MDS codes $\{\mathds{C}_t\}_{t=1}^{n-k}$, respectively.

\section{Proof of Lemma \ref{Prop:MBR for zero gammac}}\label{Proof:Properties of suggested MBR codes gammac=0 case}

We first review four properties of incidence matrix $V_t$.

\begin{prop}\label{Prop:Incidence Matrix}
	The incidence matrix $V_t$ of a fully connected graph $G_t$ with $t$ vertices has the following four properties as summarized in \cite{rashmi2009explicit}:
	\begin{enumerate}[label=(\alph*)]
		\item Each element is either 0 or 1. \label{Prop:Incidence Matrix first}
		\item Each row has exactly ($t - 1$) 1's. \label{Prop:Incidence Matrix second}
		\item Each column has exactly two 1's. \label{Prop:Incidence Matrix third}
		\item Any two rows have exactly one section of 1's. \label{Prop:Incidence Matrix fourth}
	\end{enumerate}
\end{prop}

Recall that for a given codeword $\mathbf{c} = [c_1, \cdots, c_\theta]$ with $\theta = {n_I \choose 2}L$ coded symbols, node $N(l,j)$ stores the symbol $c_{(l-1){n_I \choose 2}+i}$ if and only if $V_{n_I}(j,i) = 1$. 
Note that any natural number $s \in [\theta]$ can be uniquely expressed as an $(l_0,i_0)$ pair
where
\begin{align}
s &= (l_0-1) {n_I \choose 2} + i_0, \label{Eqn:unique expression}\\
l_0 &\in \{1,2,\cdots, L\}, \nonumber\\
i_0 &\in \{1,2,\cdots, {n_I \choose 2}\} \nonumber
\end{align}
holds.
Therefore, a coded symbol $c_s = c_{(l_0-1){n_I \choose 2}+i_0}$ is stored at node $N(l_0,j)$ if and only if $V_{n_I}(j,i_0) = 1$. From Proposition \ref{Prop:Incidence Matrix}\textendash\ref{Prop:Incidence Matrix third}, each column of $V_{n_I}$ has exactly two $1$'s. In other words, 
\begin{equation*}
V_{n_I}(j_1,i_0) = V_{n_I}(j_2,i_0) = 1
\end{equation*}
holds for some $j_1, j_2 \in [n_I]$. 
Therefore, nodes $N(l_0,j_1)$ and $N(l_0,j_2)$ store the coded symbol $c_s$. Note that no other nodes can store $c_s$ since (\ref{Eqn:unique expression}) is the unique expression of $s$ into $(l_0,i_0)$ pair. This proves Lemma \ref{Prop:MBR for zero gammac}\textendash \ref{Prop:MBR for zero gammac_first}. Note that the two nodes, $N(l_0,j_1)$ and $N(l_0,j_2)$, which share $c_s$ are located in the same cluster $l_0$. This proves Lemma \ref{Prop:MBR for zero gammac}\textendash \ref{Prop:MBR for zero gammac_second} and Lemma \ref{Prop:MBR for zero gammac}\textendash \ref{Prop:MBR for zero gammac_third}.  
Finally, according to Proposition \ref{Prop:Incidence Matrix}\textendash\ref{Prop:Incidence Matrix second} for $t=n_I$, each row of $V_{n_I}$ has $(n_I-1)$ number of $1's$. Thus, $V_{n_I}(j,i_p)=1$ holds for some $\{i_p\}_{p=1}^{n_I-1} \subseteq \{1,2,\cdots, {n_I \choose 2}\}$. Therefore, node $N(l,j)$ contains $(n_I-1)$ coded symbols of $\{c_{(l-1){n_I \choose 2}+i_p}\}_{p=1}^{n_I-1}$.
This proves Lemma \ref{Prop:MBR for zero gammac}\textendash \ref{Prop:MBR for zero gammac_fourth}.

\section{Proof of Lemma \ref{Prop:MBR}}\label{Proof:Properties of suggested MBR codes}

Recall that the suggested coding scheme obeys the following rule: for $l \in [L]$ and $j \in [n_I]$, node $N(l,j)$ stores
\begin{itemize}
	\item $c_{i_1}$ if and only if $V_n(n_I(l-1)+j,i_1)=1$, for $i_1 \in [{n \choose 2}]$
	\item $c_{{n \choose 2} + (\chi l - \chi - l + t){n_I \choose 2}+i_2}$ 
	if and only if $V_{n_I}(j,i_2)=1$, for $i_2 \in [{n_I \choose 2}]$. This rule holds for every $t\in [\chi-1]$.
\end{itemize}

Note that the first rule deals with storing $c_s$ with $s\in S_1$, and the second rule stores $c_s$ with $s\in S_2$, where
\begin{align}
S_1 & \coloneqq	\{1,2,\cdots, {n \choose 2}\} \\
S_2 & \coloneqq \{{n \choose 2}+1,{n \choose 2}+2,\cdots, \theta\} 
\end{align}
Here, 
\begin{equation}
\theta = \binom{n}{2} + (\chi-1)\binom{n_I}{2} L,
\end{equation}
as in (\ref{Eqn:theta}). 

We first focus on the coded symbols $c_s$ for $s\in S_1$. The mathematical results are summarized in the following remark, with proofs given below. 


\begin{remark}\label{Rmk:MBR_simple}
	Consider coded symbols $c_s$ for $s\in S_1$ only. Then,
	\begin{itemize}
		\item Each coded symbol is stored in exactly two different storage nodes.
		\item Nodes in different clusters share one coded symbol.
		\item Nodes in the same cluster share one coded symbol.
		\item Each node contains $n-1$ coded symbols. 
	\end{itemize}
\end{remark}
\begin{proof}
	The first statement is directly obtained from Proposition \ref{Prop:Incidence Matrix}\textendash\ref{Prop:Incidence Matrix third}, while the second and third statements are obtained from Proposition \ref{Prop:Incidence Matrix}\textendash\ref{Prop:Incidence Matrix fourth}. Finally, the last statement is from Proposition \ref{Prop:Incidence Matrix}\textendash\ref{Prop:Incidence Matrix second}. 
\end{proof}

Now we focus on the coded symbols $c_s$ for $s\in S_2$. First, note that we can represent $s = s' + {n \choose 2}$ for $s' \in S_2'$ where 
\begin{equation}
S_2' \coloneqq \{1, 2, \cdots,  (\chi-1) {n_I \choose 2} L \}.
\end{equation}
Moreover, $s' \in S_2'$ can be uniquely represented as $(l,t,i_2)$ tuple, through the following steps. Also refer to Fig. \ref{Fig:S_2_prime}.

\begin{enumerate}
	\item Divide $S_2'$ into $L$ partitions $P_1, P_2, \cdots, P_L$, where each partition $P_l$ has
	\begin{equation}
	\Delta \coloneqq (\chi-1){n_I \choose 2}
	\end{equation}
	elements. 
	To be specific, the $L$ partitions are
	\begin{align*}
	P_1 &= \{1, 2, \cdots, \Delta\}, \\
	P_2 &= \{\Delta + 1, \Delta + 2, \cdots, 2\Delta\}, \\
	& \vdots  \\
	P_L &= \{(L-1)\Delta + 1, (L-1)\Delta + 2, \cdots, L\Delta\}.
	\end{align*}
	For each $s'\in S_2'$, we can uniquely assign $l$, the index of partition which includes $s'$. For example, since $2\Delta \in P_2$, we assign $l=2$ to $s' = 2\Delta$.		 
	\item Consider a specific $P_l$ of size $\Delta$. Divide it into $(\chi-1)$ partitions $P_{l,1}, P_{l,2}, \cdots, P_{l,\chi-1}$, where each partition $P_{l,t}$ has 
	\begin{equation}
	\delta \triangleq {n_I \choose 2}
	\end{equation} 
	elements. To be specific, for $l \in [L]$, we have
	\begin{align*}
	P_{l,1} &= \{(l-1)\Delta + 1, \cdots, (l-1)\Delta + \delta\}, \\
	P_{l,2} &= \{(l-1)\Delta + \delta + 1, \cdots, (l-1)\Delta + 2\delta\}, \\
	& \vdots  \\
	P_{l,\chi-1} &= \{(l-1)\Delta + (\chi-2)\delta + 1, \cdots, l\Delta\}.
	\end{align*}
	For each $s'\in S_2'$, we can uniquely assign $(l,t)$, the index pair of the partition which includes $s'$. For example, since $2\Delta \in P_{2,(\chi-1)}$, we assign $(l,t)=(2,\chi-1)$ to $s' = 2\Delta$.	 
	\item Note that each $s'\in S_2'$ belongs to a specific $P_{l,t}$. Let $i_2$ be the position of $s'$ within the set $P_{l,t}$. For example, since $2\Delta$ is located in the last (\textit{i.e.}, $\delta^{th}$) element of $P_{2,(\chi-1)}$, we assign $i_2 = \delta$ to $s' = 2\Delta$. Therefore, $s'=2\Delta$ can be uniquely expressed as $(l,t,i_2)=(2,\chi-1, \delta)$  tuple. 
\end{enumerate}

\begin{figure}[!t]
	\centering
	\includegraphics[width=80mm]{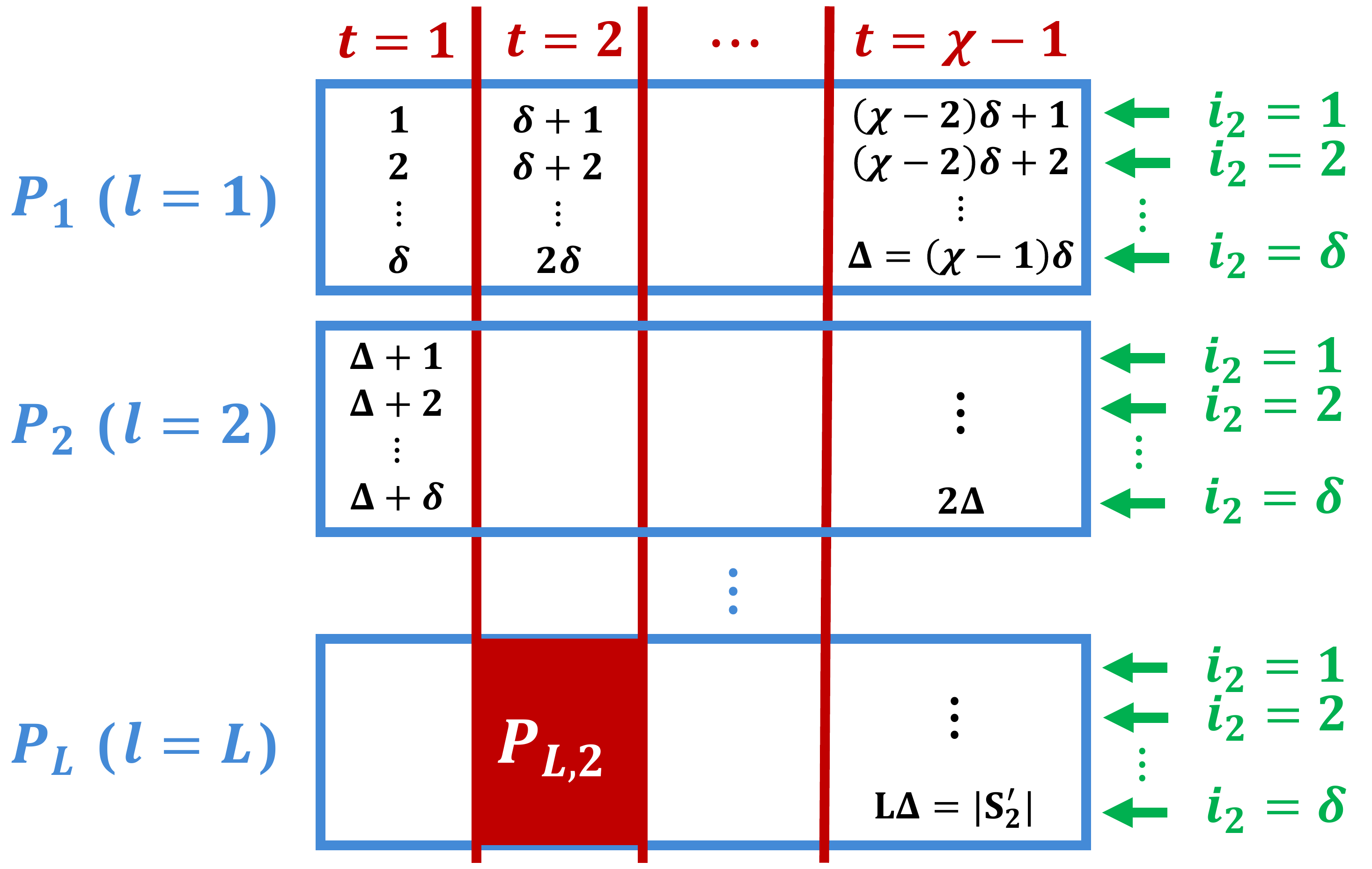}
	\caption{Representing $s' \in S_2'$ as a unique $(l,t,i_2)$ tuple}
	\label{Fig:S_2_prime}
\end{figure} 

Using the $(l,t,i_2)$ representation of $s' \in S_2'$, we can express
\begin{align}
s &= {n \choose 2} + s' \nonumber\\
&= {n \choose 2} + (l-1)(\chi-1){n_I \choose 2} + (t-1){n_I \choose 2} + i_2 \nonumber\\
&= {n \choose 2} + (\chi l-\chi-l+t){n_I \choose 2} + i_2 \label{Eqn:s_representation}
\end{align}
for $s \in S_2$. 
Note that there exists a one-to-one mapping between $s$ and $(l,t,i_2)$ tuple. Obtaining $s$ from $(l,t,i_2)$ is made clear in (\ref{Eqn:s_representation}), while the other direction is given as follows, which is from Fig. \ref{Fig:S_2_prime}.
\begin{align*}
s' &= s - {n \choose 2}, \quad l = \left\lceil{\frac{s'}{\Delta}}\right\rceil,
\quad t = \left\lceil{\frac{s'-(l-1)\Delta}{\delta}}\right\rceil, \\
i_2 &= s'-(l-1)\Delta-(t-1)\delta.
\end{align*}

Now we move onto our second remark, which is also proved. 

\begin{remark}\label{Rmk:MBR_complex}
	Consider coded symbols $c_s$ for $s\in S_2$ only. Then,
	\begin{itemize}
		\item Each coded symbol is stored in exactly two different storage nodes.
		\item Nodes in different clusters do not share any coded symbols.
		\item Nodes in the same cluster share $(\chi-1)$ coded symbols.
		\item Each node contains $(\chi-1)(n_I-1)$ coded symbols. 
	\end{itemize}
\end{remark}
\begin{proof}
	Consider an arbitrary coded symbol $c_s$ for $s \in S_2$. There exists an unique corresponding $(l,t,i_2)$ tuple. From Proposition \ref{Prop:Incidence Matrix}\textendash\ref{Prop:Incidence Matrix third}, there exists $j_1, j_2 \in [n_I]$ such that $V_{n_I}(j_1, i_2) = V_{n_I}(j_2, i_2) = 1$. Therefore, nodes $N(l,j_1)$ and $N(l,j_2)$ store $c_s$. Since the $(l,t,i_2)$ notation is unique for each $s \in S_2$, no other node can store $c_s$. This proves the first statement.  
	Moreover, since both $N(l,j_1)$ and $N(l,j_2)$ are in the $l^{th}$ cluster, each coded symbol $c_s$ is stored in exactly two different nodes in the same cluster. This proves the second statement.
	
	Consider arbitrary two nodes in the same cluster, denoted as $N(l,j_1)$ and $N(l,j_2)$. From Proposition \ref{Prop:Incidence Matrix}\textendash\ref{Prop:Incidence Matrix fourth}, there exists a unique $i_2 \in [{n_I \choose 2}]$ such that $V_{n_I}(j_1, i_2) = V_{n_I}(j_2, i_2) = 1$. Therefore, both the $N(l,j_1)$ and $N(l,j_2)$ nodes store
	\begin{equation*}
	c_{{n \choose 2} + (\chi l - \chi - l + t){n_I \choose 2}+i_2}
	\end{equation*}
	for $t=1,2,\cdots, \chi-1$. In other words, two nodes in the same cluster share $(\chi-1)$ coded symbols. This proves the third statement.
	
	Consider an arbitrary node $N(l,j)$. From Proposition \ref{Prop:Incidence Matrix}\textendash\ref{Prop:Incidence Matrix second}, there exist $\{i_2^{(v)}\}_{v=1}^{n_I-1}$ such that $V_{n_I}(j,i_2^{(v)})=1$ for $v=1,2,\cdots, n_I-1$. Thus, node $N(l,j)$ stores
	\begin{equation*}
	c_{{n \choose 2} + (\chi l - \chi - l + t){n_I \choose 2}+i_2^{(v)}}
	\end{equation*}
	for every $t \in [\chi-1]$ and $v \in [n_I-1]$. Therefore, each node stores $(\chi-1)(n_I-1)$ coded symbols, which completes the proof for the fourth statement.
\end{proof} 

From Remarks \ref{Rmk:MBR_simple} and \ref{Rmk:MBR_complex}, we obtain Lemma \ref{Prop:MBR}.

\section{Proof of Lemma \ref{Prop:Lower_bound on the number of retrievable coded symbols for zero gammac}}\label{Proof:Lower bound on the number of retrievable coded symbols_gammac0 case}

For a given contact vector $\boldsymbol{\omega}=[\omega_1, \omega_2, \cdots, \omega_L]$, we know that $\omega_1$ nodes are contacted from the $1^{st}$ cluster, $\omega_2$ nodes are contacted from the $2^{nd}$ cluster, and so on. Moreover, the total number of contacted nodes is $\sum_{l=1}^L \omega_l =k$ from the definition of the contact vector. From Lemma \ref{Prop:MBR for zero gammac}\textendash \ref{Prop:MBR for zero gammac_second} and Lemma \ref{Prop:MBR for zero gammac}\textendash \ref{Prop:MBR for zero gammac_third}, we have
\begin{equation} \label{Eqn:number of retrieved symbols for zero gammac}
n(\boldsymbol{\omega}) = k\alpha - \sum_{l=1}^L {\omega_l \choose 2}
\end{equation}
for a given $\boldsymbol{\omega} \in \Omega$.
Here, the first term represents the total number of coded symbols retrieved from $k$ nodes, each containing $\alpha$ symbols. The second term represents the number of symbols which are retrieved twice.

Consider $\boldsymbol{\omega}^* = [\omega_1^*,\cdots, \omega_L^*]$ where 
\begin{equation}\label{Eqn:optimal_a}
\omega_i^* = 
\begin{cases}
n_I,  & i \leq  	\lfloor k/n_I\rfloor \\
k-\left\lfloor\dfrac{k}{n_I}\right\rfloor n_I,  & i = \lfloor k/n_I\rfloor + 1 \\
0, & otherwise.
\end{cases}
\end{equation}
Given the sequence $\omega_1^*,\cdots, \omega_L^*$, let $\omega_{(i)}^*$ be the permuted sequence such that 
\begin{equation*}
\omega_{(1)}^* \geq \omega_{(2)}^* \geq \cdots \geq \omega_{(L)}^*
\end{equation*}
holds. Moreover, for a given arbitrary $\boldsymbol{\omega} = [\omega_1, \cdots, \omega_L] \in \Omega$, define $\omega_{(i)}$ as the permuted sequence such that
\begin{equation*}
\omega_{(1)} \geq \omega_{(2)} \geq \cdots \geq \omega_{(L)}. 
\end{equation*}
Then,  we have
\begin{align}
\sum_{i=1}^{L} \omega_i^* &= \sum_{i=1}^{L} \omega_i = k, \label{Eqn:sum is const}\\
\sum_{i=1}^{t} \omega_{(i)}^* &\geq \sum_{i=1}^{t} \omega_{(i)}  \text{  for } t=1,2,\cdots,L. \nonumber
\end{align} 
In other words, for arbitrary  $ \boldsymbol{\omega} \in \Omega$, we can conclude that $\boldsymbol{\omega}^*$ majorizes $\boldsymbol{\omega} $ (the definition of majorization is in \cite{vaidyanathan2010signal}), which is denoted as 
\begin{equation} \label{Eqn:majorization}
\boldsymbol{\omega}^* \succ \boldsymbol{\omega}.
\end{equation}
Note that $g(x) \coloneqq x^2$ is convex for real number $x$. Then, from Theorem 21.3 of \cite{vaidyanathan2010signal}, 
\begin{equation*}
f(\boldsymbol{\omega}) \coloneqq \sum_{i=1}^{L} g(\omega_i) =  \sum_{i=1}^{L} \omega_i^2 
\end{equation*}
is a schur-convex function on $\mathbb{R}^L$, where $\mathbb{R}$ is the set of real numbers. 
From the definition of Schur-convexity (definition 21.4 of \cite{vaidyanathan2010signal}), $\mathbf{x} \succ \mathbf{y}$ implies $f(\mathbf{x}) \geq f(\mathbf{y})$. Thus, from (\ref{Eqn:majorization}), 
\begin{equation*}
f(\boldsymbol{\omega}^*) = \sum_{i=1}^L (\omega_i^*)^2 \geq \sum_{i=1}^L (\omega_i)^2 = f(\boldsymbol{\omega})
\end{equation*}
holds for arbitrary $\boldsymbol{\omega} \in \Omega$. 
Therefore, 
\begin{align}\label{Eqn:majorization_result}
\boldsymbol{\omega}^* &=\aggregate{argmax}{\boldsymbol{\omega}\in\Omega}  \sum_{i=1}^L (\omega_i)^2 = \aggregate{argmax}{\boldsymbol{\omega}\in\Omega}  \sum_{i=1}^L \frac{\omega_i (\omega_i - 1)}{2}
\end{align}
where the last equality is from (\ref{Eqn:sum is const}).
From (\ref{Eqn:number of retrieved symbols for zero gammac}), we have
\begin{align*}
n(\boldsymbol{\omega}) &= k\alpha - \sum_{i=1}^L \frac{\omega_i (\omega_i - 1)}{2} \geq k\alpha - \sum_{i=1}^L \frac{\omega_i^* (\omega_i^* - 1)}{2} \nonumber\\
&= k\alpha - \frac{1}{2} ( {\boldsymbol{\omega}^*}^T \boldsymbol{\omega}^* - k).
\end{align*}
Using 
(\ref{Eqn:optimal_a}) and (\ref{Alpha_MBR_GammaC0}), this can be reduced as	
\begin{align}
n(\boldsymbol{\omega}) &\geq k (n_I-1) - \frac{1}{2} (q n_I^2 + r^2 - k) \nonumber\\
&= kn_I - \frac{1}{2}(q n_I^2 + r^2 + k). \label{Eqn:optimal_a_choose_2}
\end{align}
Moreover, using the definition in \eqref{Eqn:h_i}, we have
\begin{equation}\label{Eqn:h_i_expand}
h_i =
\begin{cases}
1, & \text{ if } i \in [g_1],\\
2, & \text{ if } i-g_1 \in [g_2],\\
\vdots & \\
n_I, & \text{ if } i- \sum_{l=1}^{n_I-1} g_l \in [g_{n_I}].
\end{cases}
\end{equation}
From \eqref{Alpha_MBR_GammaC0}, \eqref{Eqn:Capacity for gamma_c = 0} and \eqref{Eqn:h_i_expand}, the file size can be expressed as
\begin{align}
\mathcal{M} &= \sum_{i=1}^k (n_I-h_i)=\sum_{l=1}^{n_I} g_l (n_I-l) = kn_I - \sum_{l=1}^{n_I} l g_l \nonumber\\
&= kn_I - \frac{1}{2}(q n_I^2 + r^2 + k) \label{Eqn:file_size_zero_epsilon}
\end{align}
where the last two equalities are from \eqref{Eqn:Property1} and \eqref{Eqn:Property2}. Combining \eqref{Eqn:optimal_a_choose_2} and \eqref{Eqn:file_size_zero_epsilon}, we have $n(\boldsymbol{\omega}) \geq \mathcal{M}$ for all $\boldsymbol{\omega} \in \Omega$, which completes the proof.

\section{Proof of Lemma \ref{Prop:Lower_bound on the number of retrievable coded symbols}}\label{Proof:Lower bound on the number of retrievable coded symbols}
For a given contact vector $\boldsymbol{\omega}=[\omega_1, \omega_2, \cdots, \omega_L]$, we know that $\omega_1$ nodes are contacted in the $1^{st}$ cluster, $\omega_2$ nodes are contacted in the  $2^{nd}$ cluster, and so on. Moreover, the total number of contacted nodes is $\sum_{l=1}^L \omega_l =k$ from the definition of the contact vector. 
From Lemma \ref{Prop:MBR}\textendash\ref{Prop:MBR_second} and Lemma \ref{Prop:MBR}\textendash\ref{Prop:MBR_third}, we obtain 
\begin{equation} \label{Eqn:number of retrieved symbols}
n(\boldsymbol{\omega}) = k\alpha - {k \choose 2}- (\chi-1) \sum_{l=1}^L {\omega_l \choose 2},
\end{equation}
for a given $\boldsymbol{\omega} \in \Omega$.
Here, the first term represents the total number of coded symbols retrieved from $k$ nodes, each containing $\alpha$ symbols. Since any two distinct nodes share one coded symbol, we subtract the second term. Moreover, since nodes in the same cluster share $(\chi-1)$ extra symbols, we subtract the third term. 

In Appendix \ref{Proof:Lower bound on the number of retrievable coded symbols_gammac0 case}, it has been shown that the vector $\boldsymbol{\omega}^{*}$ defined in (\ref{Eqn:optimal_a}) satisfies equation (\ref{Eqn:majorization_result}), which says
\begin{equation*}
\boldsymbol{\omega}^* = \aggregate{argmax}{\boldsymbol{\omega}\in\Omega}  \sum_{i=1}^L {\omega_i \choose 2}.
\end{equation*}
Therefore, combining with (\ref{Eqn:number of retrieved symbols}), we have 
\begin{align*}
n(\boldsymbol{\omega}) &= k\alpha - {k \choose 2}- (\chi-1) \sum_{l=1}^L {\omega_l \choose 2}\\
&\geq  k\alpha - {k \choose 2}- (\chi-1) \sum_{l=1}^L {\omega_l^* \choose 2} \\
& = k\alpha - {k \choose 2}- \frac{1}{2}(\chi-1) (q n_I^2 + r^2 - k)= \mathcal{M}
\end{align*}
for all $\boldsymbol{\omega} \in \Omega$, where the second last equality is from (\ref{Eqn:optimal_a}), and the last equality is from (\ref{Eqn:Capacity for MBR}). 
This completes the proof of Lemma \ref{Prop:Lower_bound on the number of retrievable coded symbols}.

\section{Proof of Propositions}

\subsection{Proof of Proposition \ref{Prop:MBR_for_nonzero_beta_c}}\label{Proof:Capacity for MBR}

We begin with three properties, which help proving Proposition \ref{Prop:MBR_for_nonzero_beta_c}. Here, we use several definitions: $g_i, q$ and $r$ are defined in (\ref{Eqn:g_m}), (\ref{Eqn:quotient}), and (\ref{Eqn:remainder}), respectively. 

\textit{Property 1}: 
\begin{equation}\label{Eqn:Property1}
\sum_{i=1}^{n_I}g_i = k.
\end{equation}
\begin{proof}
	Note that 
	\begin{equation}\label{Eqn:g_i_brand_new}
	g_i = 
	\begin{cases}
	q+1, & i \leq r \\
	q, & \text{otherwise}
	\end{cases}
	\end{equation}
	Therefore, 
	\begin{equation*}
	\sum_{i=1}^{n_I}g_i = (q+1)r + q(n_I-r) = r + qn_I = k,
	\end{equation*}
	where the last equality is from (\ref{Eqn:remainder}).
\end{proof}

\textit{Property 2}: 
\begin{equation}\label{Eqn:Property2}
\sum_{i=1}^{n_I}ig_i = \frac{1}{2} (qn_I^2+r^2+k).
\end{equation}
\begin{proof}
	From (\ref{Eqn:g_i_brand_new}), 
	\begin{align*}
	\sum_{i=1}^{n_I}ig_i &= \sum_{i=1}^{r}(q+1)i + \sum_{i=r+1}^{n_I}qi = q \sum_{i=1}^{n_I}i + \sum_{i=1}^{r}i\\
	&= q \frac{n_I(n_I+1)}{2} + \frac{r(r+1)}{2} \\
	&= \frac{1}{2} (qn_I^2+r^2+qn_I+r)=	\frac{1}{2} (qn_I^2+r^2+k)
	\end{align*}
	where the last equality is from (\ref{Eqn:remainder}).
\end{proof}

\textit{Property 3}: 
\begin{equation}\label{Eqn:Property3}
\sum_{i=1}^{n_I}\sum_{j=1}^{g_i}\{ \sum_{m=1}^{i-1}g_m + j \} = \frac{k}{2} + \frac{k^2}{2}.
\end{equation}
\begin{proof}
	\begin{align*}
	(LHS) &= \sum_{i=1}^{n_I} \{g_i \sum_{m=1}^{i-1}g_m + \frac{g_i(g_i+1)}{2}\} \\
	&= \frac{1}{2} \{2\sum_{i=1}^{n_I}\sum_{m=1}^{i-1}g_ig_m +\sum_{i=1}^{n_I}g_i^2\} + \frac{1}{2} \sum_{i=1}^{n_I}g_i \\
	&= \frac{1}{2} \{2\sum_{i=1}^{n_I}\sum_{m=1}^{i-1}g_ig_m +\sum_{i=1}^{n_I}g_i^2\} + \frac{k}{2} = \frac{k^2}{2} + \frac{k}{2},
	\end{align*}
	where the second-last equality is from (\ref{Eqn:Property1}), and the last equality is from 
	\begin{equation*}
	k^2 = \{\sum_{i=1}^{n_I}g_i\}^2 = \sum_{i=1}^{n_I}g_i^2 + 2\sum_{i=1}^{n_I}\sum_{m=1}^{i-1}g_ig_m.
	\end{equation*}
\end{proof}

Note that we have 
\begin{equation}
\rho_i \beta_I + (n-\rho_i - (\sum_{m=1}^{i-1}g_m) - j) \beta_c \leq \gamma,\quad  \quad \forall i \in [n_I], \forall j \in [g_i]  
\end{equation}
according to Proposition 2 of \cite{sohn2018capacity}. 
Using $\alpha = \gamma$ from \eqref{Eqn:resource_for_epsilon_nonzero}, 
the capacity expression in (\ref{Eqn:Capacity of clustered DSS_rev}) reduces to
\begin{equation}\label{Eqn:capacity_simple}
\mathcal{M} = \sum_{i=1}^{n_I} \sum_{j=1}^{g_i} ( \rho_i\beta_I + (n-\rho_i - (\sum_{m=1}^{i-1}g_m) - j) \beta_c).
\end{equation}
Combining (\ref{Eqn:gamma}) and (\ref{Eqn:capacity_simple}), we have
\begin{align*}
\mathcal{M} &= \sum_{i=1}^{n_I} \sum_{j=1}^{g_i} \{\gamma - (i-1)\beta_I  - (j - i + \sum_{m=1}^{i-1}g_m ) \beta_c \}.
\end{align*}
Since $\beta_I = \chi$, $\beta_c = 1$ and $\alpha = \gamma$ from \eqref{Eqn:resource_for_epsilon_nonzero}, the capacity expression reduces to
\begin{align*}
\mathcal{M} = 
\sum_{i=1}^{n_I} & \sum_{j=1}^{g_i} \{\alpha - \chi(i-1)  + i -j - \sum_{m=1}^{i-1}g_m  \}\\
= \sum_{i=1}^{n_I} & (\alpha + \chi)g_i - \sum_{i=1}^{n_I} (\chi-1)ig_i\\
&- \sum_{i=1}^{n_I} \sum_{j=1}^{g_i} \{j + \sum_{m=1}^{i-1}g_m \}. 
\end{align*}
Using (\ref{Eqn:Property1}), (\ref{Eqn:Property2}), and (\ref{Eqn:Property3}), this in turn reduces to
\begin{align*}
\mathcal{M} = 
k&(\alpha+\chi) - (\chi-1)\frac{1}{2}(qn_I^2 + r^2 + k) - (\frac{k}{2} + \frac{k^2}{2})  \\
= k&\alpha + (\chi-1)k + k \\
& - (\chi-1)\frac{1}{2}(qn_I^2 + r^2 + k) - (\frac{k}{2} + \frac{k^2}{2}) \\
= k&\alpha - (\chi-1)\frac{1}{2}(qn_I^2 + r^2 - k) + k - (\frac{k}{2} + \frac{k^2}{2})\\
= k&\alpha - (\chi-1)\frac{1}{2}(qn_I^2 + r^2 - k) - {k \choose 2},
\end{align*}
which completes the proof.

\subsection{Proof of Proposition \ref{Prop:parameter_for_small_epsilon}}\label{Section:proof_of_prop_param_small_epsilon}
From Corollary 2 of \cite{sohn2018capacity}, the MSR point for $\epsilon = 0$ is
\begin{equation}\label{Eqn:MSR_point_review}
(\alpha, \gamma) = 
\left(\frac{\mathcal{M}}{\tau + \sum_{i=\tau + 1}^{k} z_i}, \frac{\mathcal{M}}{\tau + \sum_{i=\tau + 1}^{k} z_i} \frac{\sum_{i=\tau+1}^{k} z_i}{s_{\tau}} \right),
\end{equation}
where 
\begin{align}
\tau &= \max \{ t \in \{0,1,\cdots, k-1\} : z_t \geq 1 \}, \label{Eqn:tau}\\
z_t &= n_I - h_t, 
\label{Eqn:z_t}\\
h_t &= \min \{s \in [n_I] : \sum_{l=1}^s g_l \geq t \}, \label{Eqn:h_t}\\
s_t &= \frac{\sum_{i=t+1}^{k}z_i}{(n_I-1)+\epsilon(n-n_I)}, 
\label{Eqn:s_t}
\end{align}
where $g_l$ is in \eqref{Eqn:g_m}.
Note that 
\begin{align}\label{Eqn:tau_sum}
\tau + \sum_{i=\tau+1}^k z_i = k-q,
\end{align}
which can be proved as follows. 

First, from \eqref{Eqn:h_t}, we have 
\begin{equation}
\begin{cases}
h_t = n_I, & \quad \quad t > \sum_{l=1}^{n_I-1} g_l = k - g_{n_I} = k - q,\\
h_t < n_I, & \quad \quad  1 \leq t \leq k-q.
\end{cases}
\end{equation}
Thus, from \eqref{Eqn:z_t}, we have 
\begin{equation}
\begin{cases}
z_t = 0, & \quad \quad k - q < t \leq k,\\
z_t \geq 1, & \quad \quad  1 \leq t \leq k-q,
\end{cases}
\end{equation}
which results in 
\begin{equation}\label{Eqn:tau_result}
\tau = k-q.
\end{equation}
Since $z_i = 0$ for $i \geq \tau+1$ from the definition of $\tau$, we directly obtain \eqref{Eqn:tau_sum} from \eqref{Eqn:tau_result}.

Next, from the definition of $s_{t}$ in \eqref{Eqn:s_t} and the setting of $\epsilon=0$, we have
\begin{equation}\label{Eqn:s_tau}
s_{\tau} = \frac{\sum_{i=\tau+1}^{k} z_i}{n_I-1}.
\end{equation}
Combining \eqref{Eqn:MSR_point_review},\eqref{Eqn:tau_sum} and \eqref{Eqn:s_tau} obtains \eqref{Eqn:parameters_for_small_epsilon}. Moreover, since $\beta_c=0$, the equation \eqref{Eqn:gamma} for $\gamma$ reduces to  $\gamma = (n_I-1)\beta_I + (n-n_I)\beta_c = (n_I-1)\beta_I$. Note that $\gamma = (n_I-1)\alpha$ holds from \eqref{Eqn:parameters_for_small_epsilon}. Thus, we have $\beta_I=\alpha$, which completes the proof.

\subsection{Proof of Proposition \ref{Prop:parameter_for_large_epsilon}}\label{Section:proof_of_prop_param_large_epsilon}

We consider the $\beta_c = 1$ case without losing generality. This implies
\begin{equation}\label{Eqn:beta_I_large_epsilon}
\beta_I = 1/\epsilon
\end{equation}
according to the definition $\epsilon=\beta_c/\beta_I$. 
From Corollary 2 of \cite{sohn2018capacity}, the MSR point for $\frac{1}{n-k} \leq \epsilon \leq 1$ is given as
\begin{equation}\label{Eqn:MSR_point_review2}
(\alpha, \gamma) = (\frac{\mathcal{M}}{k}, \frac{\mathcal{M}}{k}\frac{1}{s_{k-1}}),
\end{equation}
where 
\begin{align}\label{Eqn:s_seq}
s_{k-1} &= \frac{(n-k)\epsilon}{(n_I-1) + \epsilon (n-n_I)} = \frac{n-k}{\frac{n_I-1}{\epsilon} + n-n_I}
\end{align}
from the definition of $\{s_i\}$ in \cite{sohn2018capacity}. 
Combining (\ref{Eqn:MSR_point_review2}) and (\ref{Eqn:s_seq}) results in (\ref{Eqn:parameters_for_large_epsilon}).
Note that $\gamma$ in (\ref{Eqn:gamma}) can be expressed as
\begin{align}
\gamma &= (n-n_I)\beta_c + (n_I-1)\beta_I = (n-n_I) + (n_I-1)/\epsilon , \label{Eqn:gamma_epsilon}
\end{align}
where the last equality holds due to \eqref{Eqn:beta_I_large_epsilon}.
Combining \eqref{Eqn:parameters_for_large_epsilon} and (\ref{Eqn:gamma_epsilon}), we obtain
\begin{equation*}
\gamma = \frac{\mathcal{M}}{k}\frac{\gamma}{n-k},
\end{equation*}
which gives
\begin{equation}\label{Eqn:Capacity_MSR_point}
\mathcal{M} = k(n-k).
\end{equation}
Using $\alpha=\mathcal{M}/k$ in (\ref{Eqn:MSR_point_review2}),
we have
\begin{equation}
\alpha = n-k.
\end{equation}
This completes the proof.

\bibliographystyle{IEEEtran}
\bibliography{IEEEabrv,TIT_Code}

\renewenvironment{IEEEbiography}[1]
{\IEEEbiographynophoto{#1}}
{\endIEEEbiographynophoto}

\begin{IEEEbiography}{Jy-yong Sohn}
	(S'15) received the B.S. and M.S. degrees in electrical engineering from the Korea Advanced Institute of Science and Technology (KAIST), Daejeon, Korea, in 2014 and 2016. He is currently pursuing the Ph.D. degree in KAIST. His research interests include coding for distributed storage/computing, distributed learning and information theory. He received the KAIST EE Best Research Achievement Award in 2018, the IEEE International Conference on Communications (ICC) Best Paper Award in 2017, and the Qualcomm Innovation Award in 2015.
\end{IEEEbiography}

\begin{IEEEbiography}{Beongjun Choi}
	(S'17) received the B.S. and M.S. degrees in mathematics and electrical engineering from the Korea Advanced Institute of Science and Technology (KAIST), Daejeon, Korea, in 2014 and 2017. He is currently pursuing the electrical engineering Ph.D degree in KAIST. His research interests include blockchain, coding for distributed storage system and information theory. He is a co-recipient of the IEEE international conference on communications (ICC) best paper award in 2017.
\end{IEEEbiography}

\begin{IEEEbiography}{Jaekyun Moon}
	(F'05) received the Ph.D degree in electrical and computer engineering at Carnegie Mellon University, Pittsburgh, Pa, USA. He is currently a Professor of electrical engineering at KAIST. From 1990 through early 2009, he was with the faculty of the School of Electrical and Computer Engineering at the University of Minnesota, Twin Cities. He consulted as Chief Scientist for DSPG, Inc. from 2004 to 2007. He also worked as Chief Technology Officer at Link-A-Media Devices Corporation. His research interests are in the area of channel characterization, signal processing and coding for data storage and digital communication. Prof. Moon received the McKnight Land-Grant Professorship from the University of Minnesota. He received the IBM Faculty Development Awards as well as the IBM Partnership Awards. He was awarded the National Storage Industry Consortium (NSIC) Technical Achievement Award for the invention of the maximum transition run (MTR) code, a widely used error-control/modulation code in commercial storage systems. He served as Program Chair for the 1997 IEEE Magnetic Recording Conference. He is also Past Chair of the Signal Processing for Storage Technical Committee of the IEEE Communications Society. He served as a guest editor for the 2001 IEEE JSAC issue on Signal Processing for High Density Recording. He also served as an Editor for IEEE TRANSACTIONS ON MAGNETICS in the area of signal processing and coding for 2001-2006. He is an IEEE Fellow.
\end{IEEEbiography}\vfill

\end{document}